\documentclass{article}[11pt]

\makeatletter
\usepackage[letterpaper, margin=1in]{geometry}
\usepackage[utf8]{inputenc}
\usepackage[T1]{fontenc}
\usepackage{lmodern}
\usepackage[british]{babel}
\usepackage[autostyle]{csquotes}
\usepackage[numbers]{natbib}
\usepackage{mathtools}
\usepackage{amsmath,amsthm,amssymb,amstext,amsbsy} 
\usepackage{color}
\usepackage{latexsym,graphicx}
\usepackage{subfig}
\usepackage[colorlinks,citecolor=blue]{hyperref}
\usepackage{thmtools,thm-restate}
\usepackage{xspace}
\usepackage[colorinlistoftodos]{todonotes}

\newtheorem{theorem}{Theorem}[section]
\newtheorem{claim}[theorem]{Claim}
\newtheorem{proposition}[theorem]{Proposition}
\newtheorem{lemma}[theorem]{Lemma}
\newtheorem{defn}[theorem]{Definition}

\newcommand{\FS}{\mathfrak{S}}
\newcommand{\CV}{\mathcal{V}}
\newcommand{\ip}[2]{\langle #1,#2\rangle}

\newcommand{\BR}{\mathbb{R}}
\newcommand{\BE}{\mathbb{E}}
\newcommand{\BP}{\mathbb{P}}

\newcommand{\ind}{\mathbf{1}}
\newcommand{\RN}{\mathbb{N}}
\newcommand{\RZ}{\mathbb{Z}}

\newcommand{\CH}{\mathcal{H}}
\newcommand{\CD}{\mathcal{D}}
\newcommand{\CI}{\mathcal{I}}

\newcommand{\fhat}{\widehat{f}}

\newcommand{\polylog}{\mathrm{polylog}}
\newcommand{\poly}{\mathrm{poly}}
\newcommand{\infnorm}[1]{\left\| #1 \right\|_{\infty}}

\newcommand{\disc}{\mathsf{disc}}
\newcommand\R{\mathbb{R}}
\newcommand\E{\mathbb{E}}
\newcommand*\bh{\ensuremath{\boldsymbol{h}}}
\newcommand*\bj{\ensuremath{\boldsymbol{j}}}
\newcommand*\bk{\ensuremath{\boldsymbol{k}}}
\newcommand*\bl{\ensuremath{\boldsymbol\ell}}
\newcommand*\bm{\ensuremath{\boldsymbol{m}}}
\newcommand{\sign}{\mathsf{sign}}
\newcommand{\CB}{\mathcal{B}}
\newcommand{\valuation}{\mathbf{v}}
\newcommand{\envy}{\ensuremath{\mathsf{envy}}\xspace}
\newcommand{\calA}{\ensuremath{\mathcal{A}}}
\newcommand{\boldp}{\ensuremath{\mathsf{p}}}
\newcommand{\boldx}{\ensuremath{\mathbf{x}}}

\renewcommand{\chi}{\varepsilon}

\newcommand{\ignore}[1]{}
\newcounter{note}[section]
\newcommand{\snote}[1]{\refstepcounter{note}$\ll${\bf Sahil~\thenote:}
  {\sf \color{red}  #1}$\gg$\marginpar{\tiny\bf SS~\thenote}}
\newcommand{\mnote}[1]{\refstepcounter{note}$\ll${\bf Makrand~\thenote:}
  {\sf \color{red}  #1}$\gg$\marginpar{\tiny\bf MS~\thenote}}

\title{Online Vector Balancing and Geometric Discrepancy}
\author{Nikhil Bansal\thanks{CWI Amsterdam and TU Eindhoven, \texttt{N.Bansal@cwi.nl}. Supported by the ERC Consolidator Grant 617951 and the NWO VICI grant 639.023.812.} \and
Haotian Jiang\thanks{Paul G. Allen School of CSE, University of Washington, \texttt{jhtdavid@cs.washington.edu}.
Supported in part by the National Science Fundation, Grant Number CCF-1749609, CCF-1740551, DMS-1839116.
} \and
Sahil Singla\thanks{Institute for Advanced Study and Princeton University, \texttt{singla@cs.princeton.edu}. Supported in part by the Schmidt Foundation. Part of the work done while visiting CWI Amsterdam.} \and 
Makrand Sinha\thanks{CWI Amsterdam, \texttt{makrand@cwi.nl}. Supported by the Netherlands Organization for Scientific Research, Grant Number 617.001.351 and the NWO VICI grant 639.023.812.}
}
\date{}

\begin{document}

\maketitle

%\pagenumbering{roman}
\begin{abstract}
\medskip

    We consider an {online vector balancing} question where $T$ vectors, chosen from an {arbitrary} distribution over $[-1,1]^n$, arrive one-by-one and must be immediately given a $\pm$ sign. The goal is to keep the {discrepancy}---the $\ell_{\infty}$-norm of any signed prefix-sum---as small as possible. A concrete example of this question is the {online interval discrepancy} problem where $T$ points are sampled one-by-one uniformly in the unit interval $[0,1]$, and the goal is to immediately color them $\pm$ such that every sub-interval remains always nearly balanced. As random coloring incurs $\Omega(T^{1/2})$ discrepancy, while the worst-case offline bounds are $\Theta(\sqrt{n \log (T/n)})$ for vector balancing and $1$ for interval balancing, 
    a natural question is whether one can (nearly) match the offline bounds in the online setting for these problems. One must utilize the stochasticity as in the worst-case scenario it is known that discrepancy is $\Omega(T^{1/2})$ for any online algorithm.
    
    \medskip

    In a special case of online vector balancing, Bansal and Spencer~\cite{BansalSpencer-arXiv19} recently show an $O(\sqrt{n}\log T)$ bound when each coordinate is \emph{independently} chosen. When there are \emph{dependencies} among the coordinates, as in the interval discrepancy problem, the problem becomes much more challenging, as evidenced by a recent work of Jiang, Kulkarni, and Singla~\cite{JiangKS-arXiv19} that gives a non-trivial $O(T^{1/\log\log T})$ bound for  online interval discrepancy. Although this beats random coloring, it is still far from the offline bound.
    
    \medskip
    
    In this work, we introduce a new framework that allows us to handle online vector balancing even when the input distribution has {dependencies} across coordinates. In particular, this lets us obtain a $\poly(n, \log T)$ bound for online vector balancing under arbitrary input distributions, and a $\polylog (T)$ bound for online interval discrepancy. Our framework is powerful enough to capture other well-studied geometric discrepancy problems; e.g., we obtain a $\poly(\log^d (T))$ bound for the online $d$-dimensional Tusn\'ady's problem. All our bounds are tight up to polynomial factors. 
    
    \medskip
    
A key new technical ingredient in our work is an \emph{anti-concentration} inequality for sums of pairwise uncorrelated random variables, which might also be of independent interest.

\end{abstract}

\clearpage

\setcounter{tocdepth}{2}
\tableofcontents

\clearpage

%\pagenumbering{arabic}
%\setcounter{page}{1}

\section{Introduction}

Consider the following online vector balancing question, originally proposed by Spencer \cite{Spencer77}: vectors $v_1,v_2,\ldots,v_T  \in [-1,1]^n$ arrive online, and upon the arrival of $v_t$, a sign $\chi_t \in \{\pm 1\}$ must be chosen irrevocably, so that the $\ell_\infty$-norm of the signed sum
$d_t = \chi_1 v_1 + \ldots + \chi_t v_t$  remains as small as possible. That is, find the smallest $B$
such that  $\max_{t \in [T]} \|d_t\|_{\infty}  \le B$. As we shall see later,
the problem arises naturally in various contexts where one wants to divide an incoming stream of objects, so that the split is as even as possible along each of the various dimensions that one might care about. 

A na\"ive algorithm is to pick each sign $\chi_t$ randomly and independently, which by standard tail bounds gives $B = \Theta((T \log n)^{1/2})$ with high probability.
In most of the interesting settings, $T \gg n$, and 
a natural question is whether the dependence on $T$ can be improved from $T^{1/2}$ to say, $\log{T}$, or removed altogether (possibly with a worse dependence on $n$).

\paragraph{Offline setting.} The offline version of the problem, where the vectors $v_1,\ldots,v_T$ are given in advance and the goal is to minimize $\max_{t \in [T]} \|d_t\|_{\infty}$, is known as the signed-series problem. It was first studied by Spencer~\cite{Spencer77}, who obtained a bound independent of $T$, but exponential in $n$. This was later improved by B\'ar\'any and Grinberg~\cite{BaranyGrinberg81} to $B\leq 2n$. Chobanyan~\cite{Chobanyan94} showed a beautiful connection between the signed-series problem and the classic Steinitz problem on the rearrangement of vector sequences---any upper bound on $B$ also holds for the latter problem. Steinitz problem has a much longer history, originating from a question of Riemann and L\'evy in the 19th century (c.f. the survey~\cite{Barany08} for some fascinating history).
A long-standing conjecture for both the problems, still open, is that $B = O(n^{1/2})$. Another notable bound is due to Banaszczyk~\cite{B12}, who showed that $B = O((n\log T)^{1/2})$. While the original argument in~\cite{B12} was non-constructive, a polynomial time algorithm to find such a signing was recently given in~\cite{BansalG17}. 

%In the case when one only cares about minimizing $\|d_T\|_{\infty}$ as opposed to $\max_{t \in [T]} \|d_t\|_{\infty}$, then a polynomial time algorithm is known~\cite{BansalG17,BansalDGL18}. \snote{I think Bansal-Garg does the general signed-series as well. See their Theorem 4. Open questions are in the $\ell_2$ case?}

\ignore{\color{blue}\paragraph{Offline setting.} The offline version of the problem, where the vectors $v_1,\ldots,v_T$ are given in advance, and the goal is to minimize $\max_{t \in [T]} \|d_t\|_{\infty}$, originated from a question of Riemann and L\'evy (c.f. the survey~\cite{Barany08} for some fascinating history), and has been extensively studied in discrepancy theory as the signed-series problem. 
In an influential result~\cite{Steinitz-16}, in 1916 Steinitz showed that $B \le 2n$, \snote{check if this was Steinitz, as initial results might be exp in $n$?} independent of $T$---this bound was later improved to $n$ \cite{GS80}. \snote{Again, not sure if $n$ is known beyond Steinitz, based on \cite{B12}?}
A long-standing conjecture, still open, is that $B = O(n^{1/2})$. Another notable bound is due to Banaszczyk~\cite{B12}, who showed that $B = O((n\log T)^{1/2})$. While the original argument in~\cite{B12} was non-constructive, a polynomial time algorithm to find such a signing was given recently in~\cite{BansalG17}.
}

In general, there has been extensive work on various offline discrepancy problems over last several decades, and several powerful techniques such as the partial coloring method \cite{Spencer85} and convex geometric methods \cite{Giannopoulos,Banaszczyk-Journal98,B12,MNT14} have been developed, which significantly improve upon the bounds given by random coloring.
While these initial methods were mostly non-algorithmic, several new algorithmic techniques and insights have been developed in recent years \cite{Bansal-FOCS10, Lovett-Meka-SICOMP15, Rothvoss14, EldanS18, BansalDG16, LevyRR17, BansalDGL18,DNTT18}.

\paragraph{Online setting.}
The online setting was first studied in the 70's and 80's, but it did not receive much interest later as it was realized that the best guarantees are already achieved by trivial algorithms. In particular, the $T^{1/2}$ dependence on $T$ achieved by random coloring cannot be improved \cite{Spencer77}. 
See \cite{Spencer-Book87,Barany79} for even more specific lower bounds. The difficulty is that the all-powerful adversary, upon seeing the signs chosen by the algorithm until time $t-1$, can choose the next input vector $v_t$ to be \emph{orthogonal} to $d_{t-1}$. Now, irrespective of the choice of the sign $\chi_t$, the resulting signed sum $d_t$ satisfies 
\begin{equation}
    \label{l:ortho}
\|d_t\|^2_2 ~=~ \| d_{t-1} + \chi_t v_t\|_2^2 ~=~ \|d_{t-1}\|^2_2 + 2 \chi_t  \ip{d_{t-1}}{v_t} +  \|v_t\|_2^2 ~=~  \|d_{t-1}\|^2_2 + \|v_t\|_2^2.
\end{equation}
For any $d_{t-1}$, one can always pick $v_t$ with\footnote{For any $d \in \R^n$, any basic feasible solution to $\ip{d}{x}=0$ with $x \in [-1,1]^n$ has at least $n-1$ coordinates $\pm 1$.} 
$\|v_t\|_\infty \leq 1$ and $ \|v_t\|_2^2 \geq n-1$, resulting in 
$\|d_t\|_2^2 \geq (n-1)t$, and hence $\|d_t\|_\infty = \Omega(t^{1/2})$ for all $t \in [T]$ (as long as $n>1$).

It is therefore natural to ask if relaxing the power of the adversary, or making additional assumptions on the input sequence, can lead to interesting new ideas and to algorithms that perform much better,
and in particular, give bounds that only mildly depend on $T$.

A natural assumption is that of \emph{stochasticity}: if the arriving vectors are chosen in an i.i.d. manner from some distribution $\boldp$, can we maintain that the $\ell_{\infty}$ norm of the current signed-sum $d_t$---henceforth, referred to as discrepancy---is $\poly(n)$ or $\poly(n, \log T)$?
% if the signs are being chosen online? 

\paragraph{Previous work and challenges.}
Recently, this stochastic setting was studied by Bansal and Spencer~\cite{BansalSpencer-arXiv19}, where they considered the case where $\boldp$ is the uniform distribution on all $\{-1,1\}^n$ vectors. 
They give an online algorithm achieving a bound of $O(\sqrt{n})$ on the \emph{expected} discrepancy, matching the best possible offline bound, and an $O(\sqrt{n}\log T)$ discrepancy bound at all times $t \in [T]$, with high probability.

In general, the algorithmic discrepancy approaches  developed in the last decade do not seem to help in the online setting. This is because in the offline setting, the algorithms can ensure that the discrepancy stays low by \emph{simultaneously} updating the colors of various elements in a correlated way.
In the online setting, however, 
%is that the the increase in discrepancy depends on the arriving $v_t$---
the discrepancy must necessarily rise (in the $\ell_2$ sense) whenever the incoming vector $v_t$ is almost orthogonal to $d_{t-1}$, which can happen quite often. The only thing that the online algorithm can do is to \emph{actively} try to \emph{cancel} this increase, whenever possible, by choosing the sign $\chi_t$ cleverly.
%\nbnote{Rewrote the above para a bit, to the point across more clearly.} \snote{I edited a bit further.}

%This is because the
%discrepancy always rises (in the $\ell_2$ sense) whenever the incoming $v_t$ is almost orthogonal to $d_{t-1}$, which happens reasonably often. 
%\nbnote{Is it clearer now?} \snote{Add something like: While in the offline case the algorithms ensure the discrepancy never rises much.}\mnote{always rises in one case vs never rises much in one case seems like it's only a matter of qunatitative bounds, not qualitative. I tried to make it clearer...}

%On the other hand, all known algorithms for offline discrepancy are based on
%updating the signs $\chi_t$ in some correlated way, so that the overall discrepancy {\em added} as the elements gets signed from $0$ to $\pm 1$  remains bounded.\todo{Clarify}

%\paragraph{Dependent vs. independent coordinates.}
The algorithm of \cite{BansalSpencer-arXiv19} crucially uses that if the coordinates of $v_t$ are independently distributed and mean-zero\footnote{Note that this holds in the  case of uniform distribution over $\{-1,1\}^n$.}, then for any $d_{t-1}$ the incoming vector $v_t$ will typically be far from being orthogonal to $d_{t-1}$.
More quantitatively, the \emph{anti-concentration} property for independent random variables gives that for any $d_{t-1}=(d_1,\ldots,d_n)$, the random vector $v_t=(X_1,\ldots,X_n)$ with $X_1,\ldots,X_n$ being independent and mean-zero satisfies   
\[  \E_v\Big[ |\ip{d_{t-1}}{v_t}| \Big]  = \Omega\left(\Big(\sum_{i=1}^n d_i^2 \cdot \E[X_i]^2\Big)^{1/2}\right).\] 
Whenever $|\ip{d_{t-1}}{v_t}|$ is large, the algorithm can choose $\chi_t$ appropriately to create a {\em negative drift} in \eqref{l:ortho}, to offset the increase due to the $\|v_t\|^2$ term. We give a more detailed description below in \S\ref{sec:BansalSpencer}.

In many interesting settings, however, the $X_i$'s can be \emph{dependent}. %\todo{this transition is abrupt and only makes sense after the next sentence ends, maybe add a phrase here}\nbnote{ok now?}
For example, motivated by an envy minimization problem, Jiang, Kulkarni, and Singla~\cite{JiangKS-arXiv19} considered the following natural online {interval discrepancy} problem: points $x_1,\ldots,x_T$ arrive uniformly in the interval $[0,1]$, and the goal is to assign them signs online to minimize the discrepancy of every sub-interval of $[0,1]$. (For adversarial arrivals, \cite{JiangKS-arXiv19} show $\poly(T)$ lower bounds.)
Viewing the sub-intervals (after proper discretization) as coordinates, this becomes a stochastic online vector balancing problem, but where the random variables $X_i$ corresponding to the various sub-intervals are dependent (details in \S\ref{sec:overviewInterval}).
They give a non-trivial algorithm that achieves $T^{1/\log \log T}$ discrepancy, which is much better than the $T^{1/2}$ bound obtained by random coloring, but still substantially worse than $\text{polylog}(T)$. 

%By standard reductions, it suffices to minimize the discrepancy of $\Theta(T)$ \emph{dyadic} intervals. Viewing each of these intervals as a coordinate, this is a vector balancing problem in $\Theta(T)$ dimensions. Now, under the input distribution the different coordinates will be dependent as these intervals overlap. 

% Jiang \emph{et al.} give an algorithm that achieves discrepancy $T^{1/\log \log T}$ for every interval. They also show a comparatively much smaller lower bound\footnote{Note that for the offline problem the discrepancy is at most one.} of $(\log T)^{1/4}$.
 %even though the points $x_t$ are uniform in $[0,1]$, 
 % (intervals with endpoints $k2^{-j}$ for integers $j,k$)
%\todo{Some context to transition between the paragraphs is missing -- tried to fix it}
%\nbnote{Shortened and reorganized it. Better now?}

In general,
the difficulty with dependent coordinates $X_i$ is that even a small correlation can destroy anti-concentration, which makes it difficult to create a negative drift.
For example, suppose the distribution $\boldp$ is mostly supported on vectors with an equal number of $+1$ and $-1$ coordinates. 
Now if $d$ has the form $d= c(1,\ldots,1)$, then the incoming vector $v_t$ is almost always orthogonal to it, and $\|d_T\|_2$ can potentially increase as fast as $\Omega(T^{1/2})$.

In this paper, we focus on the stochastic setting where the coordinates have {dependencies}, and give several results both for specific geometric problems and for general vector balancing under arbitrary distributions. In general, there are various other ways in which one can relax the power of the adversary, and in \S\ref{sec:open} we describe several interesting open questions and directions in this area. 

%In general, this is a widely open area with several interesting questions. We explore some of

\subsection{Our Discrepancy Bounds}

We first consider the following interval discrepancy problem. Let $x = x_1, \ldots, x_T$ be a sequence of points drawn uniformly in $[0,1]$ and let $\chi_1, \ldots, \chi_T \in \{\pm 1\}$ be a signing. For an interval $I \subseteq [0,1]$, let $\ind_{I}$ denote the indicator function of the interval $I$. For any time $t \in [T]$, we define the discrepancy of interval $I$ to be
\[  \disc_t(I) := \Big|\chi_1 \ind_I(x_1) + \cdots + \chi_t \ind_I(x_t) \Big| . \]
 We show the following bounds on discrepancy. % Throughout, we will omit the input $x$ and the signing $\chi$ when there is no ambuiguity.

\begin{theorem}[Interval Discrepancy]
    \label{thm:interval1d}
	There is an online algorithm which 
	%for points  $x=x_1,\ldots,x_T$  sampled independently and uniformly from $[0,1]$ 
	selects signs $\chi_t \in \{\pm 1\}$ such that,
	with high probability\footnote{Throughout the paper, ``with high probability'' means with $1 - 1/\poly(n,T)$ probability where the exponent of the polynomial can be made as large as desired, depending on the constant in the discrepancy upper bound.}, \emph{for every interval} $I \subseteq [0,1]$ we have $\max_{t \in [T]} \disc_t(I) = O(\log^3 T)$.
	Moreover, 
	%this bound is tight up to polynomial factors, since
	with constant probability, 
	for any online algorithm, $ \max_{I \subseteq [0,1]}  \max_{t \in [T]} \disc_t(I) = \Omega\left(\sqrt{ \log T}\right)$. 
\end{theorem}
This gives an exponential improvement over the $T^{1/\log \log T}$ bound of \cite{JiangKS-arXiv19}, and is tight up to polynomial factors. The lower bound also improves a previous bound of $\Omega(\log^{1/4}T)$ of \cite{JiangKS-arXiv19}. 
%and also provides a better lower bound. 

There are two natural $d$-dimensional generalizations of the interval discrepancy problem, and our framework, which we will describe in \S\ref{sec:framework}, can handle both of them. 

%We first define these generalizations.
% problem is trivial for $d=1$ as the  discrepancy is at most $1$ for any placement of points. 

\medskip
\emph{$d$-dimensional Online Interval Discrepancy}:
 Consider a sequence of points $x_1, \ldots, x_T$ drawn uniformly from the unit cube $[0,1]^d$. The goal is to simultaneously minimize the discrepancy of every interval for all the $d$-coordinates. In other words, to minimize the following for every interval $I$ and every coordinate $i \in [d]$:
\[ \disc^i_t(I) := \Big|\chi_1 \ind_I(x_1 (i)) + \ldots + \chi_t  \ind_I(x_t (i)) \Big|.\]
 %\medskip

 The offline version of this problem for $d \ge 2$ is equivalent to the classic $d$-permutations problem, where an upper bound of $O(\sqrt{d}\log T)$~\cite{SpencerST-SODA97} and a breakthrough lower bound of $\Omega(\log T)$~\cite{NewmanNN-FOCS12,Franks18} for $d\geq 3$, and $\Omega(\sqrt{d})$ in general is known for the worst-case placement of points.
 
We show the following generalization of Theorem~\ref{thm:interval1d} that matches the best offline bounds, up to polynomial factors.
 \begin{theorem}[$d$-dimensional Interval Discrepancy]
    \label{thm:interval}
	There is an online algorithm which 
	%given points  $x=x_1,\ldots,x_T$  sampled independently and uniformly from $[0,1]^d$
	selects signs $\chi_t \in \{\pm 1\}$ such that,
	with high probability, $\text{ for each } i\in [d] \text{ and } I \subseteq [0,1]$, we have $\max_{t \in [T]} \disc_t^i(I) = O(d \log^3 T)$.
Moreover, 
%this bound is tight up to polynomial factors since
with constant probability, for any online algorithm
there exists an interval $I$ and a coordinate $i \in [d]$, such that
$ \max_{t \in [T]} \disc_t^i(I) \allowbreak = \Omega\big(\sqrt{d \log \left(T/d\right)}\big)$. 
\end{theorem}
Previously, Jiang et al.~\cite{JiangKS-arXiv19} could  extend their analysis for online interval discrepancy to the $d=2$ case and prove the same $T^{1/\log\log T}$ bound. However, their proof is rather ad-hoc and does not seem to generalize to higher $d$.  In contrast, our bound holds for any $d$, and is tight up to polynomial factors.

 The second natural generalization of interval discrepancy is to $d$-dimensional axis-parallel boxes, which gives the following online version of the extensively studied Tusn\'ady's Problem.

\emph{$d$-dimensional Online Tusn\'ady's Problem}: Consider a sequence of points $x_1, \ldots, x_T$ drawn uniformly from the unit cube $[0,1]^d$. The goal is to simultaneously minimize the discrepancy of all axis-parallel boxes. In other words, to minimize the following for every box $B$:
\[ \disc_t(B) := \Big|\chi_1 \ind_B(x_1) + \ldots + \chi_t  \ind_B(x_t)\Big|.\]
%\medskip

%Previously, Jiang, Kulkarni and Singla \cite{JiangKS-arXiv19} proved a sub-polynomial $T^{1/\log \log T}$ bound for the above problem when $d\le 2$ but their results do not extend to higher dimensions and also to setting of Tusn\'ady's problem discussed below.
%\vspace{0.2pt}

The (offline) Tusn\'ady's problem has a fascinating history (see~\cite{Matousek-Book09} and references there in), and after a long line of work, it is known that for the worst-case placement of points, the offline discrepancy is at most $O_d(\log^{d - \frac12} T)$ \cite{Nikolov-Mathematika19}  and at least $\Omega_d(\log^{d-1} T)$~\cite{MN-SoCG15}. 
We show the following result in the online setting, which is tight to within polynomial factors.

\begin{theorem}[Tusn\'ady's problem]
    \label{thm:tusnady}
	There is an online algorithm which 
	%given points  $x=x_1,\ldots,x_T$  sampled independently and uniformly from $[0,1]^d$ 
	selects signs $\chi_t \in \{\pm 1\}$ such that,
	with high probability, for every axis-parallel box $B$, we have $\max_{t \in [T]} \disc_t(B) = O_d(\log^{2d+1} T)$.
	Moreover,
	%this bound is tight up to polynomial factors since with constant probability, 
	for any online algorithm, with constant probability, there exists a box $B$ such that
	$\max_{t \in [T]} \disc_t(B) = \Omega_d(\log^{d/2} T)$.
\end{theorem}
In contrast, the proof approach of \cite{JiangKS-arXiv19} completely breaks down for the Tusn\'ady's problem even in two dimensions and does not give any better lower bounds in terms of $d$.
We recently learned that results similar to Theorems \ref{thm:interval1d} and \ref{thm:tusnady} were also obtained by Dwivedi et al.~\cite{DFGR19}, in the context of understanding the power of online thinning in reducing discrepancy.
 
 {\bf Remark:} Although all the problems above are stated for uniform distributions, one can use the probability integral transformation to reduce any product distribution to the uniform distribution without increasing the discrepancy, so our results in Theorems \ref{thm:interval} and \ref{thm:tusnady} also apply to any product distribution over $[0,1]^d$.

Finally, note that Theorem \ref{thm:interval1d} follows as a direct corollary of either of the above theorems.

\paragraph{General distributions.}
We now consider the setting of \emph{arbitrary} distributions for the online vector balancing problem. Here we need to tackle the orthogonality issue which gave $\Omega(T^{1/2})$ lower bounds  discussed in~\eqref{l:ortho}. 
As discussed earlier, for the uniform distribution over $\{-1,+1\}^n$, Bansal and Spencer~\cite{BansalSpencer-arXiv19} get around this issue since this does not happen for the uniform distribution reasonably often, and hence, $\E[\ip{d_{t-1}}{v_t}]$ is large for any vector $d_{t-1}$. Using this, they obtain the bound $ O(n^{1/2} \log T)$.   
Our next result shows that such a $\poly(n, \log T)$ upper bound is possible even for arbitrary distributions.

%consider the general vector balancing problem for an arbitrary input distribution and an arbitrary time horizon $T$, and show a $\poly(n, \log T)$  upper bound.
%\snote{Say more why interesting. Connect back to Bansal-Spencer.}

\begin{restatable}{theorem}{mainVector}\emph{(Vector balancing under dependencies)}
    \label{thm:signedseries}
For any sequence of vectors $v_1, \ldots, v_T \in [-1,1]^n$ sampled i.i.d. from some  arbitrary distribution $\boldp$, there is an online algorithm which selects signs $\chi_t \in \{\pm 1\}$ such that, with high probability, we have
\[ \max_{t \in [T]}~ \Big\|{\chi_1 v_1 + \ldots + \chi_t  v_t}\Big\|_{\infty}  = ~O(n^{2} (\log T + \log n)).\] 
\end{restatable}
In \S\ref{sec:lowerBounds} we show that the dependencies on $n$ and $\log T$ in this theorem are tight up to polynomial factors as there is an $\Omega(n^{1/2}+({\log T}/\log \log T)^{1/2})$ lower bound.

All of the above results follow from a general framework that we discuss next. In addition to the framework below, the key new technical ingredient is an anti-concentration inequality for dependent random variables, which we describe below in Theorem~\ref{lemma:anticonc}. This may be of independent interest.

\subsection{Our Framework}
\label{sec:framework}
 To tackle the orthogonality issue, one of our key idea is to work with a \emph{different basis} for the discrepancy vectors. More specifically, instead of maintaining bounds on the individual coordinate discrepancies $d_t(i)$, we maintain bounds on suitable linear combinations of them. 
 This basis ensures that the (new) coordinates of the incoming vector are \emph{uncorrelated}, i.e., $\BE[X(i) \cdot X(j)]=\BE[X(i)]\cdot \BE[X(j)]$ for distinct coordinates $i,j$.
 Note that this condition is only on the expected values, and is much weaker, e.g., even pairwise independence. Once one finds a suitable new basis, which turns out to be an eigenbasis of the covariance matrix, the anti-concentration bound for such random variables (proved below in Theorem~\ref{lemma:anticonc}), together with the standard exponential penalty based framework used in previous works \cite{BansalSpencer-arXiv19, JiangKS-arXiv19}, gives 
 Theorem~\ref{thm:signedseries}. %A thorough discussion of our approach is given in \S\ref{sec:overview}. 
 
 %phenomena still hold for uncorrelated mean-zero random variables. %This might be of independent interest and is discussed below in \S\ref{sec:anticoncintro}. 

%In this case, working with above is that bounding the discrepancy in a new basis preserves $\ell_2$ discrepancy in the original basis but 

%The difficult choice is to then find a suitable basis to work with. 
For our results on geometric discrepancy problems, there is an additional challenge,  we cannot afford to lose a $\poly(n)$ factor, as in Theorem~\ref{thm:signedseries} above, since the dimension $n=\Theta(T)$. In this case, however, the update vectors are $(\log T)$-sparse in the original basis (see \S\ref{sec:overview}) and one could hope to utilize this sparsity. Yet another challenge in this case is that bounding the discrepancy in a new basis preserves $\ell_2$-discrepancy in the original basis, but could lead to a $\sqrt{n}$ loss in $\ell_{\infty}$-discrepancy. To get $\polylog(T)$ bounds, we use a natural basis from wavelet theory, called the \emph{Haar system}, which simultaneously has sparsity, uncorrelation, and avoids the $\ell_2$ to $\ell_{\infty}$ loss. This also easily extends to higher dimensions as these wavelets can be tensorized in a natural way to get a suitable basis for higher dimensional versions of the problems. A more detailed description of our framework is given in \S\ref{sec:overview}. Next we discuss our anti-concentration results.

%one needs to find a new basis that avoids this $\ell_2$ to $\ell_{\infty}$ loss and where the update vectors in the new basis are sparse and uncorrelated. 

%\todo{Moved the Beck-Fiala stuff to open problems section.}

\subsection{Our Anti-Concentration Results for Non-Independent Random Variables}
\label{sec:anticoncintro}
Suppose $X_1,\ldots,X_n$ are independent $\{-1,+1\}$ random variables with mean zero. Then, it is well-known that $|\sum_i X_i|$ has mean $\Theta(n^{1/2})$, and moreover, this value is at least $\Omega(n^{1/2})$ with constant probability. 

Now, on the other hand, consider the following distribution. 
Let $H_n$ be $n \times n$ Hadamard matrix and let $H_n(i)$ denote its $i$-th row for $i \in [n]$. Consider the random vector $X=(X_1,\ldots,X_n)$, where $X = \xi \cdot H_n(i)$ for a Rademacher random variable $\xi \in \{-1, +1\}$ and a uniformly chosen  $i \in [n]$. Then the $X_i$'s are still mean-zero and $\{-1,+1\}$, and in fact, pairwise independent. However, the magnitude of the sum $|\sum_i X_i|$ behaves very differently from the i.i.d. setting above. It takes value $n$ with probability only $1/n$ (if $X= \xi \cdot H_n(1)$, the row of all $1$'s) and is $0$ otherwise. In particular the mean is $\E[|\sum_i X_i|] = 1$ (instead of $n^{1/2}$ above), and moreover the entire contribution to the mean comes from an event with probability only $1/n$.

Nevertheless, we can say interesting things about the anti-concentration of sums of such random variables. In particular, we show the following results for uncorrelated  or pairwise independent random variables.

\begin{restatable}{theorem}{anticonc} \label{lemma:anticonc} \emph{(Uncorrelated anti-concentration)}
For any $(a_1, \ldots, a_n) \in \R^{n}$, let $X_1,\ldots,X_n$ be uncorrelated random variables that are bounded $|X_i| \leq c$, satisfy $\BE[X_iX_j]=0$ for all $i \neq j$, and have sparsity $s$ (the number of non-zero $X_i$'s in any outcome). Then
\begin{align} \label{eq:pairwiseUncor}
    \E \Big|\sum_i a_i X_i\Big|  ~\geq~   \E\Big[\sum_i |a_i| X_i^2\Big]\cdot \frac{1}{cs}. 
\end{align}
Moreover, this bound is tight, even for pairwise independent random variables.
\end{restatable}
The tightness holds for the Hadamard example above, where $\E |\sum_i X_i|=1$, $s=n$, $c=1$, and $\E[\sum_iX_i^2]=n$.

\begin{restatable}{theorem}{anticoncPairwise} \label{lem:pairwiseIndep} \emph{(Pairwise independent anti-concentration)}
For any $(a_1, \ldots, a_n) \in \R^{n}$, let $X_1,\ldots,X_n$  be mean-zero pairwise independent random variables with sparsity $s \leq n$. 
%(i.e., the number of non-zero $X_i$'s in any outcome). 
Then
\begin{align} \label{eq:pairwiseIndep}
    \E \Big[\Big|\sum_i a_i X_i\Big| \Big] ~\geq~   \E\Big[\sum_i |a_i X_i| \Big]\cdot \frac{1}{s}. 
\end{align}
\end{restatable}
Note that this bound is also tight for the Hadamard example. In general, the bound \eqref{eq:pairwiseIndep} is stronger than in \eqref{eq:pairwiseUncor}; and a simple example in \S\ref{sec:pairwiseIndep} shows that \eqref{eq:pairwiseIndep} cannot hold for uncorrelated random variables.

Although the anti-concentration properties and the small-ball probabilities for independent variables have been extensively studied (c.f.~\cite{NV13}), the uncorrelated and pairwise independent setting does not seem to have been studied before, and Theorems \ref{lemma:anticonc} and \ref{lem:pairwiseIndep} do not seem to be known, to the best of our knowledge.

\subsection{Applications to Envy Minimization}
\label{sec:envyintro}

A classic measure of fairness in the field of fair division is {envy}~\cite{FoleyEssay67,ThomsonVarian-Essay85,LiptonMMS-EC04,Budish-Journal11}. A recent work of Benade et al.~\cite{BenadeKPP-EC18} introduced the \emph{online envy minimization} problem where $T$ items  arrive one-by-one. In the two player setting, on arrival of item $t \in \{1, \ldots, T\}$ we get to see the valuations $v_{it} \in [0,1]$ for both the players $i\in \{1,2\}$. The goal is to immediately and irrevocably allocate the item to one of the players while  minimizing the maximum \emph{envy}. There are two natural notions of envy: cardinal and ordinal (see \S\ref{sec:envy} for  definitions). Benade et al. \cite{BenadeKPP-EC18} show an $\Omega(T^{1/2})$ lower bound for online envy minimization in the \emph{adversarial} model---the reason is similar to B\'ar\'any's~\cite{Barany79} lower bound for online discrepancy. 
Can we obtain better bounds when the player valuations are drawn from a distribution?\footnote{If we make a simplifying assumption that the distribution does not depend on the time horizon $T$, better bounds are known~\cite{ZP-arXiv19,DGKPS-AAAI14}.}

In the special case of product distributions (each player independently draws their value), Jiang et al.~\cite{JiangKS-arXiv19} observed that the $2$-dimensional interval discrepancy bounds also hold for online envy minimization. In particular, they obtained a $T^{1/\log\log T}$ bound on the ordinal envy. Our new interval discrepancy bound from Theorem~\ref{thm:interval} immediately improves this to an $O(\log^3 T)$ bound on ordinal envy. Moreover, we use our vector balancing result to obtain an $O(\log T)$ bound on the cardinal envy even for general distributions.

%Our results on vector balancing and $2$-dimensional interval discrepancy can be used to immediately get online algorithms achieving $\polylog(T)$ bounds on both cardinal and the ordinal envy under stochastic inputs.

\begin{restatable}{corollary}{envyCorol} \label{cor:envyBounds}
Suppose  valuations of two players are drawn i.i.d. from some distribution~$\boldp$ over $[0,1]\times [0,1]$. Then, for an arbitrary distribution $\boldp$ (i.e., player valuations for the same item could be correlated),  the online cardinal envy is $O(\log T)$. 
Moreover, if  $\boldp$ is a product distribution (i.e., player valuations for the same item are  independent) then the  online ordinal envy is also $O(\log^3 T)$.
\end{restatable}

\subsubsection*{Paper Organization}

The rest of the paper is organized as follows: in \S\ref{sec:overview}, we give an overview of previous challenges and our main ideas. In \S\ref{sec:anticonc}, we prove our key anti-concentration theorems that are necessary for our upper bounds on discrepancy. In \S\ref{sec:discuncor}, we give upper and lower bounds for online discrepancy under certain ``uncorrelation'' assumptions on the distribution. Then, we apply these bounds in \S\ref{sec:signedseries} to obtain our vector balancing result (Theorem~\ref{thm:signedseries}). In \S\ref{sec:geomdisc}, we again apply these bounds to obtain our geometric discrepancy results (Theorems~\ref{thm:interval} and \ref{thm:tusnady}). In \S\ref{sec:envy}, we show why our results immediately apply to online envy minimization. Finally, in \S\ref{sec:open} we end with some discussion of open problems and directions.

\section{Proof Overview}

\label{sec:overview}

Let us start by reviewing the approach considered by Bansal and Spencer~\cite{BansalSpencer-arXiv19} in the case of independent coordinates. We also discuss the challenges involved in extending it to the setting of dependent coordinates.

%-------------------------------------
\subsection{Independent Coordinates: Bansal and Spencer}
\label{sec:BansalSpencer}

Consider the online vector balancing problem, when each arriving vector is uniformly chosen from $\{\pm 1\}^n$, so that all the coordinates are independent. To design an online algorithm, it is natural to keep a potential function that keeps track of the discrepancy and chooses a sign $\chi_t$ for the current vector $v_t$ that minimizes the increase in the potential. Formally, let $d_{t} = \chi_1 v_1 + \ldots + \chi_tv_t$ denote the discrepancy vector at time $t$. For a parameter $0 < \lambda < 1$, define the potential function 
\[ \Phi_{t} = \sum_{i \in [n]} \cosh(\lambda d_{t}(i)),\]
where $d_{t}(i)$ denotes the $i$th coordinate of $d_t$ and $\cosh(x) = \frac12 \cdot ({e^{x} + e^{-x}})$ for all $x\in \BR$. One should think of the above potential function as a proxy for the maximum discrepancy as $\Phi_t$ is dominated by the maximum discrepancy: $\Phi_t \approx e^{\lambda \infnorm{d_t}}$.

%Since the update vector $v_t$ is sampled independently of all the previous one, the sequence of random variables $(\Phi_t)_{t \in [T]}$ performs a random walk.

On the arrival of  vector $v_t$, the algorithm chooses a sign $\chi_t \in \{\pm 1\}$, which updates the discrepancy vector to $d_t = d_{t-1} + \chi_t v_t$ and changes the potential from $\Phi_{t-1}$ to $\Phi_t$. If we can show that whenever $\Phi_t > 2n$, the \emph{drift} $\Delta \Phi_t := \Phi_t - \Phi_{t-1}$ is negative in expectation for the sign $\chi_t$ chosen by the algorithm, then we can say that the potential after $T$ arrivals, $\Phi_T$, is bounded by $\poly(nT)$ with high probability. This  implies $\cosh(\lambda \| d_T \|_{\infty})$ is bounded by $\poly(nT)$, which means a bound of $O(\lambda^{-1} \log T)$ on the maximum discrepancy.

Let us try to compute the expected drift. Define $d = d_{t-1}$. By considering the Taylor expansion, we get $\cosh(x+\delta) \leq \cosh(x) + \sinh(x)\delta + \cosh(x)\delta^2$  where $\sinh(x) = \frac12\cdot({e^{x} - e^{-x}})$ for all $x \in \BR$. So,
\[ \Delta \Phi_{t} ~~\approx~~ \sum_{i \in [n]} \Big(\lambda \sinh(\lambda d(i))\cdot(\chi_t v_t(i)) +  \lambda^2\cosh(\lambda d(i))\cdot(\chi_t v_t(i))^2 \Big) ~~=~~  \chi_t \lambda L +  \lambda^2 Q,\]
where $L = \sum_{i \in [n]} \sinh(\lambda d(i)) \cdot v_t(i)$ is the \emph{linear} term and $Q =  \sum_{i \in [n]} \cosh(\lambda d(i))$ is the \emph{quadratic} term from the Taylor expansion (note that $(\chi_tv_t(i))^2=1$). Since the algorithm is free to choose the sign $\chi_t$ to minimize the drift, $\Delta \Phi_{t} \approx -\lambda |L| + \lambda^2 Q$. Now if one can show that $\BE_{v_t}[|L|] \ge \frac{\BE[Q]}{2\lambda}$, we would get that the expected drift $\BE[\Delta \Phi_{t}] < 0$, and this would translate to a good discrepancy bound of $O(\lambda^{-1} \log T)$ if $\lambda$ is large as described above. 

%Note that since $v_t$ is a vector with $\pm 1$ coordinates $\BE[Q] = \Phi_{t-1} > 2n$.
%we want to relate

Since $\cosh(x)$ and $|\sinh(x)|$ only differ by at most $1$, we can make the approximation $Q \approx \sum_{i\in [n]}|\sinh(\lambda d(i))|$ up to some small error. So, denoting $\beta = 1/\lambda$ and $a_i = \sinh(\lambda d(i))$, our task reduces to proving the following anti-concentration statement:

\paragraph{Question.} Let $X_1, \ldots, X_n$ be independent random variables with $|X_i|\le 1$. What is the smallest $\beta$ such that the following holds:
\begin{equation}
    \label{eqn:anticonc}
\ \BE\Big[ \Big| \sum_{i \in [n]}  a_i X_i \Big| \Big] \ge \frac{1}{\beta} \cdot \BE\Big[ \sum_{i \in [n]}  |a_i| X_i^2 \Big].  
\end{equation}

%where the expectation on the left and right hand sides are $\BE[|L|]$ and $\BE[Q]$ respectively. 
%$\BE|\sum_{i \in [n]} \alpha_i X_i|$ and $Q = \sum_{i \in [n]} |\alpha_i| X_i|$ for independent Rademacher random variables $X_i$ where $\alpha_i = \sinh(\lambda d_{t}(i)$.

In the case where the $X_i$'s are independent Rademacher ($\pm 1$) random variables, classical Khintchine's inequality and Cauchy-Schwarz tell us that
\[\BE\Big[ \Big| \sum_{i \in [n]}  a_i X_i \Big| \Big] ~~ {\ge} ~~ \frac{1}{\sqrt{2}} \cdot \Big(\sum_{i \in [n]} a_i^2 \Big)^{1/2} ~~ {\ge} ~~ \frac{1}{\sqrt{2n}} \Big(\sum_{i \in [n]} |a_i|\Big) ~~ {=} ~~ \frac{1}{\sqrt{2n}} \cdot \BE \Big[\sum_{i \in [n]} |a_i|X_i^2  \Big],\]
so $\beta = O(\sqrt{n})$, which suffices for the discrepancy application. In general, when $X_i$'s are not Rademacher but are still bounded ($|X_i| \le 1$), mean-zero, and \emph{independent}, then following \cite{BansalSpencer-arXiv19} one can still show that $\beta = O(\sqrt{n})$.

The above gives a bound of $O(\sqrt{n}\log T)$ on the maximum discrepancy at every time $t \in [T]$. However, when the input distribution has dependencies across coordinates, i.e. the $X_i$'s are dependent, one can not take $\beta$ to be small in general. For example, $\beta \rightarrow \infty$ when all $a_i$'s are one and a random set of coordinates $S \subset [n]$ of size $n/2$  (say $n$ is even) take value $+1$ and the remaining coordinates in $[n]\setminus S$ take value $-1$. 

Next we discuss the simplest geometric discrepancy problem---the interval discrepancy problem in one dimension---where such a situation already arises if we use the same approach as above. 
%{\color{red} In fact, even in a simple setting with dependencies among coordinates---the interval discrepancy problem in one dimension---such a situation already arises if we use the same approach as before. }

%with constant probability $\left|\sum_{i \in [n]}  \alpha_i X_i \right| \ge \frac{1}{\sqrt n} \left( \sum_i |\alpha_i| \right)$ and this is enough to make the proof work. 

%Even though $\BE[L]=0$, the variance of the random variable is $\sigma^2 := \BE[L^2] = \sum_{i} a_i^2 \ge \frac{(\sum_i |a_i|)^2}{n} = \frac{(\BE[Q])^2}{n}$, so we know that it is typically not zero. If we could show an anti-concentration statement of the form that $|L| \ge \sigma/2$ with constant probability, it will be enough for rest of the proof to go through. Note that this is quite non-trivial in this setting --- because the coefficients $a_i$ can be very large, most of the contribution to $\BE[L^2]$ could come from some-low probabibility event when $L$ takes very large values, so we will not get a good bound. This is where the independence of the coordinates is very useful: one can prove in this setting that $|L| \ge \frac{\sigma}{4} = \BE[Q]/sqrt{n}$ happens with constant probability which is enough for the rest of the proof to go through.

%-------------------------------------
\subsection{Interval Discrepancy: Previous Barriers}
\label{sec:overviewInterval}

Recall, we have $T$ points $x_1, \ldots, x_T$ chosen uniformly from $[0,1]$ which need to be given $\pm 1$ signs online. Consider the \emph{dyadic} intervals $I_{j,k} := [k2^{-j},(k+1)2^{-j}]$ where $0 \le k < 2^j$ and $0 \le j \le \log T$. For intuition, imagine embedding the unit interval on a complete binary tree of height $\log T$; now sub-intervals corresponding to every node of the binary tree are dyadic intervals. Note that the smallest dyadic interval has size $2^{-\log T} = 1/T$. 
By a standard reduction, every sub-interval of $[0,1]$ is contained in a union of some $O(\log T)$ dyadic intervals, so it suffices to track the discrepancy of these dyadic intervals. 

Denoting by $\ind_I$ the indicator function for an interval $I$, define 
 \[ d_t(I) := \chi_1 \ind_{I}(x_1) + \ldots + \chi_t \ind_I(x_t).\] 
Note that $|d_t(I_{j,k})|$ is the discrepancy of the interval $I_{j,k}$ at time $t$.
A natural choice of algorithm is to  use the potential function
\[ \Phi_{t} = \sum_{j,k} \cosh(\lambda d_t(I_{j,k})),\]
which is a proxy for the maximum discrepancy of any dyadic interval. Ideally, we want to set $0< \lambda < 1$ as large as possible. Defining $d_{j,k}=d_{t-1}(I_{j,k})$, and doing a similar analysis as before, we derive
\[ \Delta \Phi_{t} \approx \chi_t\lambda L +  \lambda^2 Q,\]
where $L = \sum_{j,k} \sinh(\lambda d_{j,k}) \cdot \ind_{I_{j,k}}(x_t)$ and $Q =  \sum_{j,k} \cosh(\lambda d_{j,k}) \cdot \ind_{I_{j,k}}(x_t)^2$. The problem again reduces to showing an anti-concentration statement as in Eq.~\eqref{eqn:anticonc} with $X_i$'s   being the indicators $\ind_{I_{j,k}}$ for all $j,k$. It turns out that the smallest $\beta$ one can hope for this setting is exponential in the height of the tree (see Appendix~\ref{app:example} for an example), which for binary trees of height $\log T$ only yields a $\poly(T)$ bound on the discrepancy.

One can still leverage something out of this approach---letting $B = T^{1/\log \log T}$, it was shown by Jiang, Kulkarni, and Singla~\cite{JiangKS-arXiv19}  that by embedding $B$-adic intervals on a $B$-ary tree of height $\log\log T$, the above approach gives
a sub-polynomial $T^{1/\log\log T}$ bound for the interval discrepancy problem. However, this cannot be pushed to give a $\polylog(T)$ bound because the above obstruction does not allow us to handle trees of  height $\log T$.

%-------------------------------------------
\subsection{Interval Discrepancy: A New Potential and the BDG Inequality} \label{sec:NewPotential}

To get around the previous problem, we take a different approach and instead of directly using  the discrepancies in the potential $\Phi_t$, we work with linear combinations of discrepancies with the following desirable properties. First, if there is a bound on these linear combinations then it should imply a bound on the original discrepancies. Second, and more importantly, the term $L$ in $\Delta \Phi_t$ can be viewed as a martingale, which leads to much better anti-concentration properties, i.e., smaller $\beta$ in~\eqref{eqn:anticonc}.

More specifically, consider the previous embedding of the dyadic intervals of length at least $1/T$ on the complete binary tree of depth $\log T$. For any interval $I_{j,k}$, let the left half interval be $I_{j,k}^l$ and the right half interval be $I_{j,k}^r$, and consider the difference (see Figure~\ref{fig:bdg}) of their discrepancies \[ d_t^-(I_{j,k}) := d_t(I_{j,k}^l) - d_t(I_{j,k}^r). \] 
Note that if $|d_t(I_{j,k})| \le \alpha$ and also $|d^-_t(I_{j,k})| \le \alpha$, then both $|d_t(I_{j,k}^l)| \le \alpha$ and $|d_t(I_{j,k}^r)|\le \alpha$. A simple inductive argument now shows that if  $|d_t([0,1])| \le \alpha$ and the differences of discrepancy for every dyadic interval $I_{j,k}$  satisfies $|d_t^-(I_{j,k})| \le \alpha$, then every dyadic interval also has discrepancy at most $\alpha$, thus satisfying the first property above. 
So let us consider a different potential function:
\[ \Xi_{t} := \cosh(\lambda d_t(I_{0,0})) + \sum_{j,k} \cosh(\lambda d_t^-(I_{j,k})) \]
with $j,k$ ranging over all the dyadic intervals (corresponding to internal nodes of the tree) and $0< \lambda < 1$ is a parameter that we want to set as large as possible. Denoting ${d^-_{j,k}}={d^-_{t-1}}(I_{j,k})$, as before, we can write
$\Delta \Xi_{t} \approx \chi_t  \lambda L +  \lambda^2 Q$, with
\begin{align*}
    \ L &= \sinh(\lambda d_t(I_{0,0})) + \sum_{j,k} \sinh(\lambda d^-_{j,k}) \cdot X_{j,k}(x_t) ~\text{ and }\\
     \ Q &= \cosh(\lambda d_t(I_{0,0})) + \sum_{j,k} \cosh(\lambda d^-_{j,k}) \cdot X_{j,k}(x_t)^2,
\end{align*}
where  $X_{j,k} = \ind_{I_{j,k}^l} - \ind_{I_{j,k}^r}$
for any interval  $I_{j,k}$. Note that $X_{j,k}$ takes value $1$ on the left half of  $I_{j,k}$, and $-1$ on the right half of $I_{j,k}$, and is zero otherwise.
% \begin{align} \label{eq:defnPsi}
%     \Psi_{j,k} := \ind_{I_{j-1,k}^l} - \ind_{I_{j-1,k}^r},
% \end{align}

\begin{figure}[h!]
   \centering
   \subfloat[The discrepancy $d_{j,k}$ terms for intervals $I_{j,k}$]{{\includegraphics[width=0.4\textwidth]{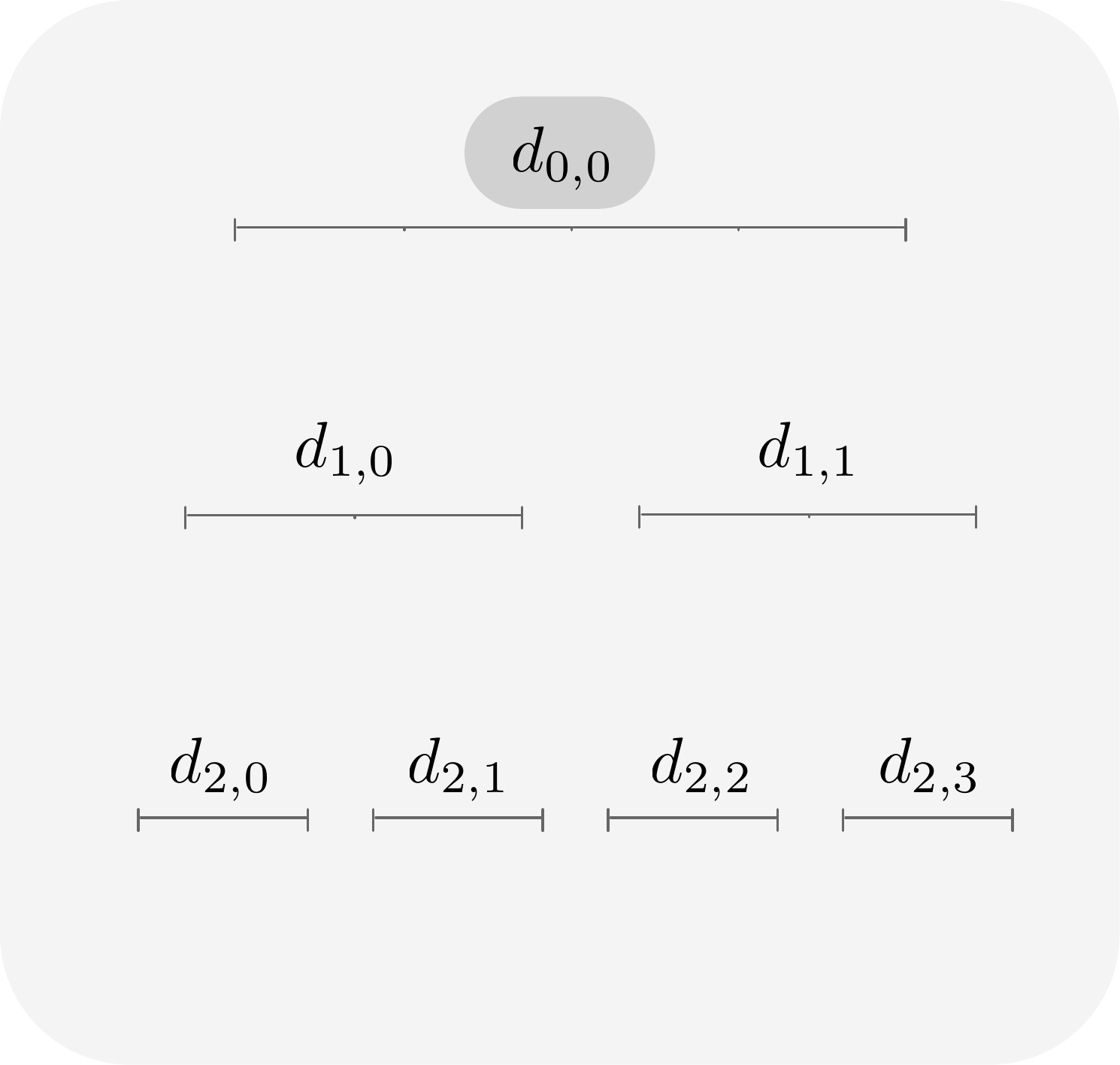}}}%
    \qquad\qquad
    \subfloat[The difference of discrepancy $d^-_{j,k}:=d_t^-(I_{j,k})$ terms for intervals $I_{j,k}$]{{\includegraphics[width=0.4\textwidth]{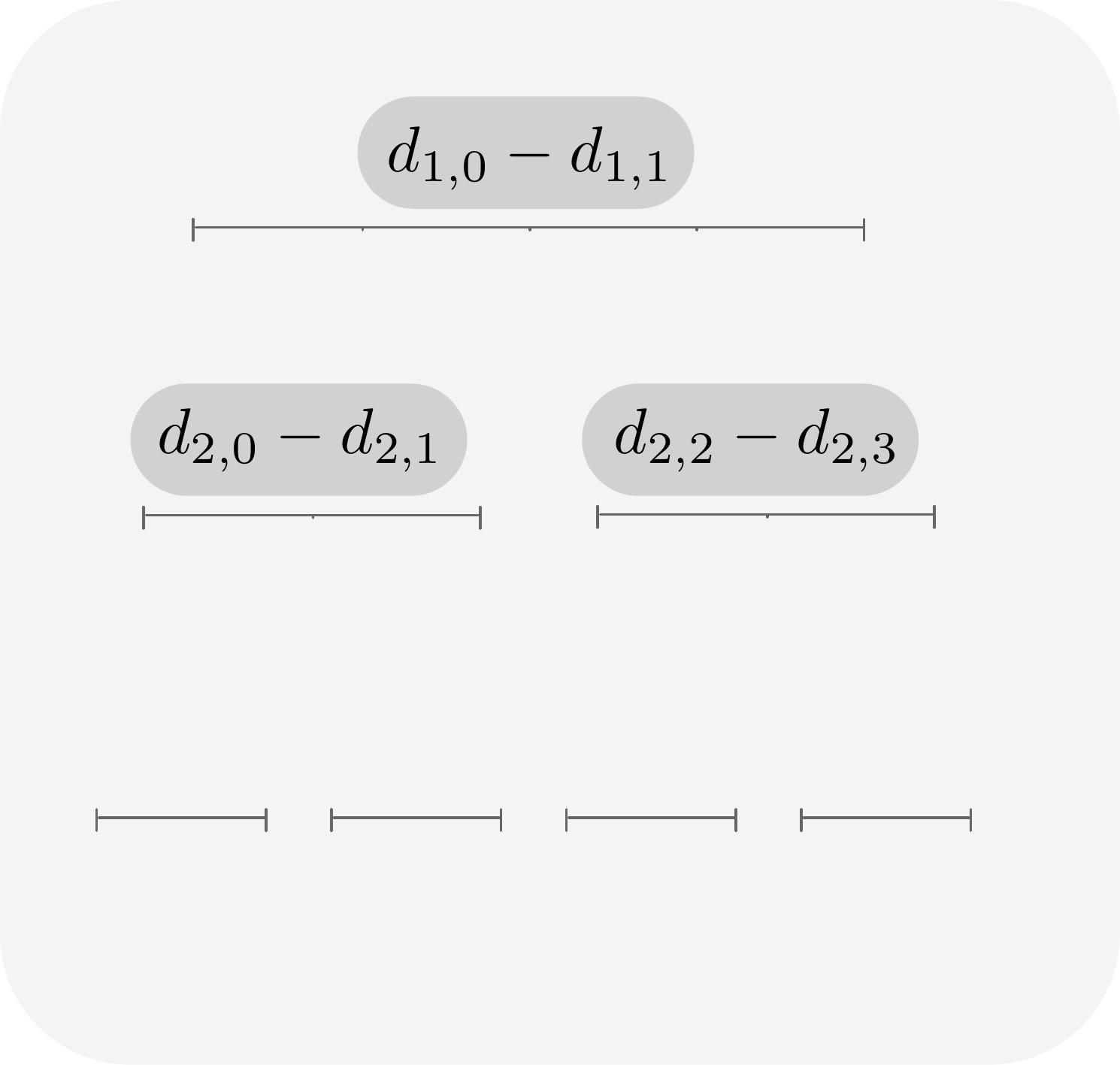}}}%
    \caption{Some terms appearing in the new potential function $\Xi_t$. Note that the hyperbolic cosine for the highlighted terms appears in $\Xi_t$.
    }%
    \label{fig:bdg}%
\end{figure}

\paragraph{Anti-concentration via Martingale analysis.} Now we show how the random variable $L$
can be viewed as a $(\log T)$-step martingale.
%value of the following -step random walk on the dyadic tree.
Let us view a uniform point $x \in [0,1]$ as being sampled one \emph{bit} at a time, starting with the most significant bit. At any point where $j$ bits of $x$ have been revealed, the interval $I_{j,k}$ on the $j^{\text{th}}$ level of the dyadic tree is determined. 
Now, consider the process that starts with the value $Y_0 = \sinh(\lambda d_{0,0})$ at the root and at any time $0 \le j \le \log T$, the process is on some node of the $j^\text{th}$ level. Conditioned on this node being $I_{j,k}$, the payoff $Y_j := a_j X_j$ where $a_j = \sinh(d^-_{j,k})$ and $X_j$ equals $1$ if the process moves to the left child and equals $-1$ otherwise. Defining $L_j = Y_0 + Y_1 \ldots + Y_j$, it follows that the sequence $L_0, \ldots, L_{\log T}$ is  a martingale and $L = L_{\log T}$. 

Moreover,  by  the approximation $\cosh(x) \approx |\sinh(x)|$, we get that $Q = |Y_0|+|Y_1|+\ldots + |Y_{\log T}|$.
Letting $a_0 = Y_0$ and $X_0=1$, the question then becomes---what is the smallest $\beta$ such that the following holds:
\[ \BE\left|\sum_{i=0}^{\log T} a_i X_i\right| ~\ge~ \frac{1}{\beta} \cdot \BE\left[\sum_{i=0}^{\log T} |a_i |X_i^2\right]  ~=~  \frac{1}{\beta} \cdot \BE\left[\sum_{i=0}^{\log T} |a_i |\right] .\]

For martingales, a statement similar to  Khintchine's inequality is implied by the well-known Burkholder-Davis-Gundy (BDG) inequality (see Theorem \ref{thm:bdg} in Appendix \ref{sec:bdg}):
\[ \BE  \left[\max_{t \le \log T} \Big|\sum_{i=0}^{t} a_i X_i\Big| \right] ~\ge~ c\cdot \BE\left[\Big(\sum_{i=0}^{\log T} a_i^2\Big)^{1/2}\right] \]
for a positive constant $c$. One can also prove (see Lemma \ref{lemma:lstar} in Appendix \ref{sec:bdg}) that 
\[  (1 + \log T) \cdot \BE \left|\sum_{i=0}^{\log T} a_i X_i\right| ~\geq~ \BE \left[ \max_{t \le \log T} \Big|\sum_{i=0}^{t} a_i X_i \Big| \right] .\] 
Then, similar to the analysis for independent Rademacher random variables, using Cauchy-Schwarz, 
\[ (1 + \log T) \cdot \BE \left|\sum_{i=0}^{\log T} a_i X_i\right| ~~ \geq~~  c\cdot \BE\left[\Big(\sum_{i=0}^{\log T} a_i^2\Big)^{1/2}\right] ~~ \ge ~~ \frac{c}{\sqrt{\log T}} \cdot \BE\left[\sum_{i=0}^{\log T} |a_i| \right].\]
So we can conclude that $\beta=\polylog(T)$, which gives a $\polylog(T)$ bound on interval discrepancy. 

How to extend this analysis to $d$-dimensional Tus\'nady's problem? The martingale analysis above strongly relied on the interval structure of the problem, which is not clear even for the two-dimensional Tus\'nady's problem. To answer this question, we take a much more general view of our online discrepancy problem.\footnote{
The more general view in fact gives a (slightly) better bound for interval discrepancy than the martingale based argument above. However, we include this martingale argument here, as it is insightful and could be useful for other problems.}
%\nbnote{Rewrote 2.3 quite a bit. Pls Check.}

\subsection{%What's Really Happening? 
A More General View of Changing Basis}

One can also view the above analysis of the interval discrepancy problem as a more general underlying principle---that of working with a different \emph{basis}. For example, let us take a linear algebraic approach to interval discrepancy and consider it as a vector balancing problem in $\R^{\CD}$, where $\CD = \{I_{j,k}~|~0\le j \le \log T, 0\le k <2^j\}$ is the set of all dyadic intervals. 
When a new point $x \in [0,1]$ arrives, the coordinate $I \in \CD$ of the update vector $v_t$ is given by
            \[ v_t(I) = \ind_{I}(x).\]
Note that the update $v_t$ lives in a $T$-dimensional subspace $\CV$ of the $(2T-1)$-dimensional space $\BR^{\CD}$ since the $T$-intervals, $I_{\log T,k}$, at the bottom layer determine the rest of the coordinates. 

The original potential function $\Phi$ from \S\ref{sec:BansalSpencer} corresponded to working with the original basis, but with the potential function $\Xi$ from \S\ref{sec:NewPotential}, our approach consisted of bounding the $\ell_{\infty}$-discrepancy in a different basis of the subspace $\CV$. In general, we may choose any basis and then define a potential function as the sum of hyperbolic-cosines of the coordinates. To choose the right basis, we   need several properties from it, but most importantly we need uncorrelation.

%To choose the right basis, we need two properties: (1) anti-concentration in the new basis, and (2) be able to translate the discrepancy bound in the new basis back to the original one.

\paragraph{Uncorrelation and anti-concentration via the Eigenbasis.}
Recall that we say random variables $X,Y$ are \emph{uncorrelated} if $\BE[XY]=\BE[X]\cdot\BE[Y]$, which is a condition only on the expected values of the random variables. Using Theorem \ref{lemma:anticonc}, to show anti-concentration  it suffices that the coordinates in the \emph{new basis} are mean-zero and uncorrelated, i.e., $\BE_v[v(i)v(j)]=0$ for distinct coordinates $i, j$.
 
%Our first technical lemma is that to show anti-concentration, akin to \eqref{eqn:anticonc}, the random variables need not be independent---it suffices that the coordinates in the new basis are mean-zero and uncorrelated, i.e., $\BE_v[v(i)v(j)]=0$ for all  $i\neq j$.  \snote{Some of this is now repetition from the intro, so we might cut.}

%\anticonc*

%Note that the squares of $X_i$'s appear in the right hand side, compared to the previous 

 For our vector balancing results under arbitrary distributions in Theorem~\ref{thm:signedseries}, we work in an \emph{eigenbasis} of the covariance matrix.  As will be shown in the proof later, standard results from linear algebra imply that the coordinates are uncorrelated in any eigenbasis. 
 Our next lemma uses  this anti-concentration (along with the hyperbolic cosine potential) to bound discrepancy in the \emph{new basis} in terms of \emph{sparsity}---number of non-zero coordinates---of the incoming vectors.

\begin{restatable}{lemma}{mainGeneralLemma}
    \label{lemma:main} \emph{(Bounded discrepancy)}
    Let $\boldp$ be a distribution supported over $s$-sparse vectors in $[-1,1]^n$ satisfying $\BE_{v \sim \boldp}[v(i)v(j)]=0$ for all $i \neq j \in [n]$. Then for  vectors $v_1, \ldots, v_T$ sampled i.i.d. from $\boldp$, there is an online algorithm that maintains  $O(s (\log n + \log T))$ discrepancy with high probability.
\end{restatable}

Even though this lemma implies low discrepancy in the new basis, we need to be careful in bounding discrepancy in the original basis.

\paragraph{Sparsity and going back to the original basis.}
 As discussed briefly in \S\ref{sec:framework}, although working in an eigenbasis allows us to obtain polynomial bounds for vector balancing,  this is apriori not sufficient for our polylogarithmic geometric discrepancy bounds. There are two main challenges---firstly, working in a new basis might lose any sparsity that we might have in the original basis; e.g., in the one-dimensional interval discrepancy problem the arriving vectors are $(\log T)$-sparse (dyadic intervals) in the original basis, but could be $\Omega(T)$-sparse in the new basis; 
and secondly, even if one can find a new basis where the coordinates are uncorrelated and have low sparsity, Lemma~\ref{lemma:main} only implies low $\ell_{\infty}$-discrepancy in the new basis. So going back to the original basis might lose us a factor $\sqrt{n}$ more (we can only claim $\ell_2$-discrepancy is the same). Recall, when we view interval discrepancy as vector balancing, $n=\Theta(T)$, so we cannot afford losing $\sqrt{n}$.  Fortunately, there is a special basis consisting of \emph{Haar wavelets} that allows us to prove $\polylog(T)$  geometric discrepancy bounds.

 %To prove our  vector balancing  results under arbitrary distributions, working in the {eigenbasis} of the covariance matrix allows us to derive $\poly(n)\cdot \log T$ bounds.
%However, this is not enough for our results on geometric discrepancy: 

%even if one can find a basis where the coordinates are uncorrelated with low sparsity, and can prove a low $\ell_{\infty}$-discrepancy bound in this basis, the $\ell_{\infty}$-discrepancy in the original basis can in general be a factor $\sqrt{n}$ more, where $n$ is the dimension of the vectors. For the geometric discrepancy problems, $n$ will be $\Theta(T)$, so this only gives a $\poly(T)$ bound while we want a $\polylog(T)$ bound. Fortunately, there is a special basis consisting of \emph{Haar wavelets} that allows us to prove $\polylog(T)$ bounds for the geometric discrepancy problems.

%-------------------------------------
\subsection{Haar Wavelets: Polylogarithmic Geometric Discrepancy}

There is a natural orthogonal basis associated with the unit interval---the basis of Haar wavelet functions. These consist of the functions $\Psi_{j,k}$'s shown in  Figure~\ref{fig:haar1d}. Together these functions are known to form an orthogonal basis for functions on the unit interval with bounded $L_2$-norm.
%constant function $\Psi_{0,0}:=\ind$ and the functions $\Psi_{j,k} := X_{j-1,k}$ defined before in~\eqref{eq:defnPsi} (also

\begin{figure}[htb]
\begin{center}
\includegraphics[width=14cm]{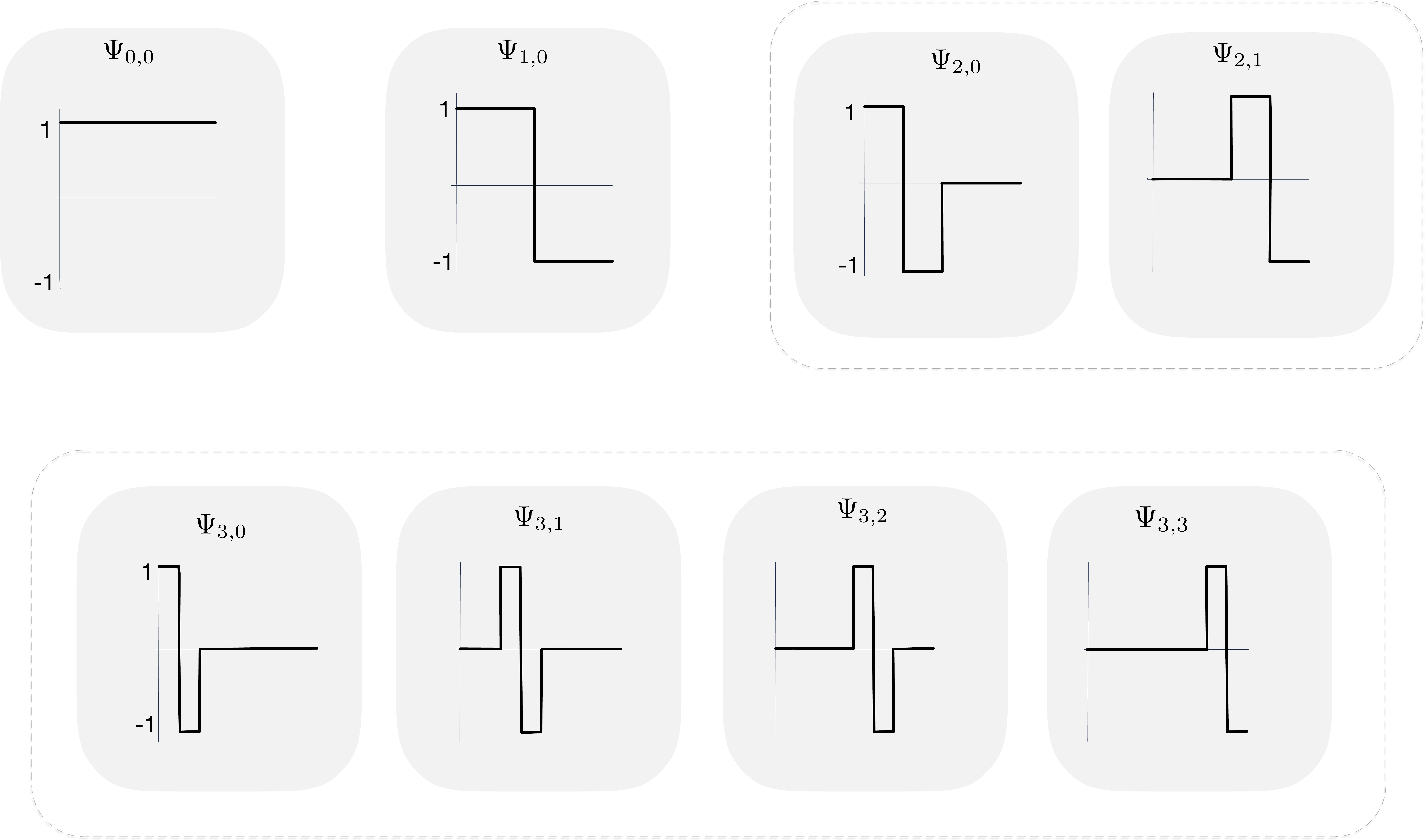}
\caption{Haar wavelets in one dimension}
\label{fig:haar1d}
\end{center}
\end{figure}

Associated with the one-dimensional Haar wavelets is a natural martingale, which is the same martingale that our previous analysis in \S\ref{sec:NewPotential} relied on (e.g., $ X_{j,k}=\Psi_{j+1,k}$ in the notation of \S\ref{sec:NewPotential}.). It turns out that the Haar wavelets have nice orthogonality and sparsity properties that allow us to use Lemma~\ref{lemma:main}---in particular, $\BE_x[h(x)h'(x)]=0$ for distinct Haar wavelet functions $h \neq h'$ and $x$ sampled uniformly from $[0,1]$. Moreover, moving from the basis of Haar wavelets to the original basis does not incur any additional loss in the discrepancy bound, since for any dyadic interval $I$, one can show that its discrepancy
\[ |d_t(I)| ~~\le~~ \alpha  \|\widehat{\ind}_I\|_1,\]
where $\alpha$ is a bound on the discrepancy in the Haar basis and $\|\widehat{\ind}_I\|_1$ is the $\ell_1$-norm of the function $\ind_I$ in the Haar basis. We prove that this $\ell_1$-norm is one, so $|d_t(I)|\le \alpha$. This gives a  more direct proof of the $\polylog(T)$ interval discrepancy bound and also extends easily to the $d$-dimensional interval discrepancy problem.

%Letting   $\widehat{1}_I(h)$ represent the coefficient of the function $1_I$ in the Haar system basis, this happens because for any dyadic interval $I$, its discrepancy 
%\[ |d_t(I)| ~~=~~ \Big|\sum_t \chi_t \cdot  1_{I}(x_t)\Big| ~~=~~ \Big|\sum_{t} \chi_t \cdot  \sum_{h} \widehat{1}_I(h) h(x_t)\Big| ~~=~~ \Big|\sum_h \widehat{1}_I(h) \cdot  \Big(\sum_t \chi_t h(x_t)\Big)\Big|,\]
%where the sum ranges over all the Haar wavelets $h$ (including the constant function $\ind$). One can verify that $\sum_t \chi_t h(x_t) = d^-_t(I_{j,k})$ if $h=\Psi_{j+1,k}$ (and if $h$ is the constant function $\sum_t \chi_t h(x_t) = d_t([0,1])$) in the notation of \S\ref{sec:NewPotential}, so if we have a bound of $\alpha$ on the magnitude of these values then the discrepancy of the dyadic interval $I$ is 
%\[ |d_t(I)| ~\le~ \alpha \Big(\sum_h |\widehat{1}_I(h)|\Big).\]

\begin{figure}[htb]
\begin{center}
\includegraphics[width=14cm]{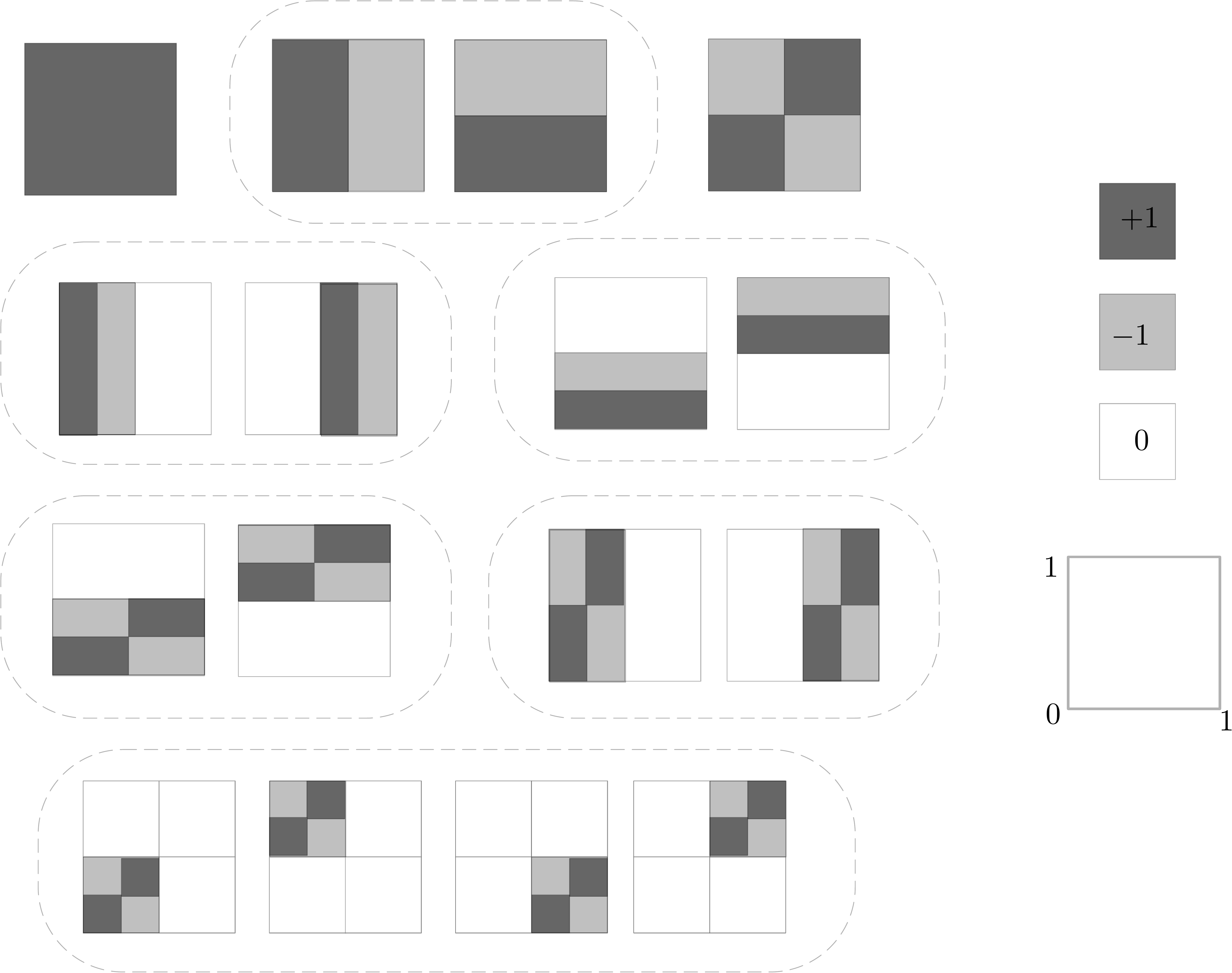}
\caption{Haar wavelets in two dimensions}
\label{fig:haar2d}
\end{center}
\end{figure}

\paragraph{Tus\'nady's problem.} Given the above framework of working in the Haar basis,  our extension to the $d$-dimensional Tus\'nady's problem now naturally follows. For example, in two dimensions, we work with the basis of Haar wavelet functions which is formed by a taking tensor product $\Psi_{j,k} \times \Psi_{j',k'}$ of the one dimensional wavelets (see Figure~\ref{fig:haar2d}). These functions form an orthogonal basis for all  bounded \emph{product} functions over $[0,1]^2$ and have nice sparsity properties. Moreover, we prove that for any axis-parallel box, the $\ell_1$-norm of the Haar basis coefficients is one, so we do not lose any additional factor in the discrepancy bound while moving from the Haar basis to the original basis. This gives a polylogarithmic bound for  two-dimensional Tus\'nady's problem, and also  extends easily to higher dimensions.

\subsection*{Notation}

All logarithms in this paper will be base two. For any integer $k$, throughout the paper  $[k]$ will denote the set $\{1, \ldots, k\}$. For a vector $u \in \BR^d$, we use $u(i)$ to denote the $i^\text{th}$ coordinate of $u$ for $i \in [d]$. Given another vector $v \in \BR^d$, the notation $u \le v$ denotes that $u(i) \le v(i)$ for each $i \in [d]$. The all ones vector is denoted by $\ind$. Given a distribution $\boldp$, we use the notation $x \sim \boldp$ to denote an element $x$ sampled from the distribution $\boldp$. For a real function $f$, we will write $\BE_{x \sim \boldp}[f(x)]$ to denote the expected value of $f(x)$ under $x$ sampled from $\boldp$. If the distribution is clear from the context, then we will abbreviate the above as $\BE_{x}[f(x)]$.

\section{Anti-Concentration Estimates}
\label{sec:anticonc}

In this section we prove the anti-concentration results: we first prove it for uncorrelated random variables, and then give an improved bound for pairwise independent random variables. Although in the rest of this paper we only use the weaker bound for uncorrelated random variables,  we think the improved anti-concentration for pairwise independent random variables is of independent interest and will find applications in the future.

\subsection{Pairwise Uncorrelated Random Variables}

The following anti-concentration bound will be used in our discrepancy applications.

\anticonc*

Note that if we have pairwise uncorrelated mean-zero random variables $X_1, \ldots, X_n$, then we get $\BE[X_iX_j]=\BE[X_i]\cdot \BE[X_j]=0$, so the above lemma implies anti-concentration in this case. The bound in the above lemma is tight  because of the Hadamard example described previously in \S\ref{sec:anticoncintro}.

%Before we present the proof of the above lemma, let us see an example which shows that this bound is tight. Without loss of generality, consider the case where $s=n$ since in the general case we can append random variables that are always zero. Let $n = 2^k$ and take $a_1=a_2=\ldots=a_n=1$. For each $b,z\in \{0,1\}^k$, define random variables 
%\[ X_b(z) = c\cdot(-1)^{\ip{b}{z}},\]
%where $\ip{b}{z} = \sum_{i=1}^n b_iz_i \bmod 2$. Note that $X_b$'s are just the scaling of the Fourier characters over $\F_2^k$, so they are pairwise independent with $\BE[X_bX_{b'}]=0$ for $b\neq b'$.  Also, it is easy to see that 
%\[ \BE\Big|\sum_b X_b\Big| = c \qquad \text{ and } \qquad  \BE\Big[\sum_b X_b^2\Big] = c^22^k = c^2 n.\]

%We prove tightness for $t=n$ as in the general case we can just append some zeros. 
%Suppose $a = (1,\ldots,1)$ and $Y = cH_n$, where $H_n$ is the Hadamard matrix.
%%
%Then the $X_i$ are infact pairwise independent.
%%
%$Ya$ has first row $cn$ and all other rows $0$. So $\E[|L|] =  c$.
%Next, as $Y^TY = c^2 n I_n$, we have  
%\[\E[Q] =  \frac{1}{n} |Y^T Ya|_1 = \frac{1}{n} c^2n I_n a =  c^2 n.\]
%

The following is the main claim in the proof of  Theorem~\ref{lemma:anticonc}. Roughly it says that $\E \Big|\sum_i a_i X_i\Big|  \geq \frac{1}{c} \cdot \max_{k \in [n]} \E[|a_k|X_k^2]$. Combined with the observation that $\max_{k \in [n]} \E[|a_k|X_k^2] \ge \frac1n \cdot \E\left[\sum_k |a_k| X_k^2\right]$ this implies Theorem~\ref{lemma:anticonc} when sparsity $s=n$. However, to get inequality \eqref{eq:pairwiseUncor} in terms of sparsity $s$, the statement of the claim has to be more refined.

\begin{claim} 
\label{claim:singlek}
For any $(a_1, \ldots, a_n) \in \R^n$ and random variables $X_1,\ldots,X_n$ satisfying $|X_i| \leq c$ and $\BE[X_iX_j]=0$ for distinct $i,j$, the following holds for any $k\in [n]$,
    \[ \E \left[\Big|\sum_i a_i X_i\Big| \cdot 1_{X_k \neq 0} \right] ~\geq~ \frac{1}{c}\cdot \E[|a_k| X_k^2].
\]
\end{claim}

\begin{proof}
Using that  $ |X_k|\leq c$, we have
\begin{align*}
    c\cdot \E \left[\Big|\sum_i a_i X_i\Big| \cdot 1_{X_k \neq 0} \right] &\geq  \E \left[\Big|\sum_i a_i X_i\Big| \cdot |X_k|  \right] \\
    \ & = \E \left[ \Big|a_k X_k^2 + \sum_{i\neq k} a_i X_i X_k\Big| \right]  ~\geq~ \E \left[ \sign(a_k) \Big(a_k X_k^2 + \sum_{i\neq k} a_i X_i X_k \Big) \right]. 
    %& = \E \Big[ |a_k| X_k^2 \Big] + sign(a_k) \cdot \E \Big[\sum_{i\neq k} a_i X_i X_k| \Big] \\
\end{align*}
    Since $\BE[X_iX_k]=0$ for $i\neq k$, it follows that
\begin{align*}
\  c\cdot \E \left[\Big|\sum_i a_i X_i\Big| \cdot 1_{X_k \neq 0} \right]  &\ge  \E \left[ |a_k| X_k^2 \right] + \sum_{i\neq k} a_i\cdot  \sign(a_k) \cdot \E \left[ X_i X_k \right] \\
    \ & =  \E \left[ |a_k| X_k^2 \right]. \qedhere
\end{align*} 
\end{proof}

When combined with the following easy claim, this will prove Theorem~\ref{lemma:anticonc}.

\begin{claim} \label{claim:CombiningYks}
Let $Y_1, \ldots, Y_n$ be correlated random variables such that for any outcome at most $s$ of them are non-zero. Moreover, suppose there is a random variable $L$ which satisfies 
    \[ \E \Big[|L|\cdot 1_{Y_k \neq 0} \Big] \geq \E\Big[|Y_k|\Big]~ \text{ for all } k \in [n].
\]
Then, $\E[|L|] \geq \frac{1}{s} \sum_k \E[|Y_k|]$.
\end{claim}
\begin{proof}
    Sum the given inequality for all $k \in [n]$ to get
\[
     \sum_k \E\Big[|Y_k|\Big] ~~ \leq ~~ \sum_k  \E \Big[|L|\cdot 1_{Y_k \neq 0} \Big] ~~= ~~  \E\Big[|L|\cdot \sum_k 1_{Y_k \neq 0} \Big] ~~\leq ~~ \E \Big[|L|\cdot s \Big] . \qedhere
\]
\end{proof}

\begin{proof}[Proof of Theorem \ref{lemma:anticonc}]
    Applying Claim \ref{claim:singlek} and Claim \ref{claim:CombiningYks} (with $L = \sum_{i} a_iX_i$ and  $Y_i = \frac{1}{c} \cdot |a_i|X_i^2$), we get that  
    \[ \E \left[\Big|\sum_i a_i X_i\Big|\right] ~~\ge~~  \E\left[\sum_k |a_k| X_k^2\right] \cdot\frac{1}{cs}. \qedhere \] 
\end{proof}

%----------------------------------------------------------

\subsection{Pairwise Independent Random Variables} \label{sec:pairwiseIndep}

In the special case of pairwise independent random variables, it is possible to obtain an improved inequality over Theorem~\ref{lemma:anticonc}. 

\anticoncPairwise*

Notice, \eqref{eq:pairwiseIndep} immediately implies \eqref{eq:pairwiseUncor}  for mean-zero pairwise independent random variables with $|X_i|\leq c$. One cannot hope to prove the stronger statement~\eqref{eq:pairwiseIndep} for uncorrelated random variables due to the following example.

\paragraph{Example.} Let $0 < \delta \ll 1$. Suppose $X_1, X_2$ are real random variables distributed over four outcomes:
\[ (X_1, X_2) = \begin{cases} \left(\dfrac{1}{\delta}, \dfrac{1}{\delta}\right) \text{ or }  \left(-\dfrac{1}{\delta}, -\dfrac{1}{\delta}\right) & \text{ w.p. } \frac{\delta^2}{2(1+\delta^2)} \text{ each}, \\
    (1,-1)  ~\text{ or }~  (-1,1) & \text{ w.p. } \frac{1}{2} - \frac{\delta^2}{2(1+\delta^2)} \text{ each}.
\end{cases} \]
Here $X_1$ and $X_2$ are uncorrelated because
\begin{align*}
\E[X_1 X_2] 
~~ = ~~ \frac{1}{\delta^2} \cdot \frac{\delta^2}{1+\delta^2} - 1 \cdot \Big(1 - \frac{\delta^2}{1+\delta^2} \Big) ~~=~~ 0 .
\end{align*}
Now it is easy to verify that $X_1$ and $X_2$ are mean zero, and \[ \BE[|X_1+ X_2|] = \frac{2 \delta}{1 + \delta^2}~~ \text{ and }~~ \BE[|X_1|+|X_2|]= \frac{2 + 2\delta}{1 + \delta^2} .\]
Therefore, the ratio between the two expectations can be made arbitrarily bad by making $\delta \to 0$.

%Consider any large integer $m$. Suppose there are two uncorrelated random variables $X_1$ and $X_2$ with mean zero having the following distributions: With probability $1/(2m^2 + 2)$ they both take value $+m$; with probability $1/(2m^2 + 2)$ they both take value $-m$; with probability $m^2/(2m^2 + 2)$ we have $X_1$ taking value $+1$ and $X_2$ taking value $-1$; finally, with the remaining probability $m^2/(2m^2 + 2)$ we have $X_1$ taking value $-1$ and $X_2$ taking value $+1$. It is easy to verify that both have mean zero and they are uncorrelated. However, $\E[|X_1+X_2|] = m/(m^2 +1) = O(1/m)$ but $\E[|X_1|+|X_2|] = 2m/(m^2 +1) + 2m^2/(m^2 +1) = \Omega(1)$.

Next, we prove Theorem~\ref{lem:pairwiseIndep}. We start with the following claim.
\begin{claim} \label{claim:pairwiseIndepHeart}
    For any $(a_1, \ldots, a_n) \in \R^{n}$ and mean-zero pairwise independent random variables $X_1,\ldots,X_n$, 
    the following holds for any $k \in [n]$, 
\[ \E \Big[\Big|\sum_i a_i X_i\Big| \cdot 1_{X_k \neq 0} \Big] ~\geq~ \E[|a_k X_k|].
\]
\end{claim}

\begin{proof}
We have
    \begin{align*}
    \E \Big[\Big|\sum_i a_i X_i\Big| \cdot 1_{X_k \neq 0} \Big] 
    &= \E \Big[\Big|a_k X_k + \sum_{i\neq k} a_i X_i \cdot 1_{X_k \neq 0} \Big| \Big]\\
    & \geq \E \Big[\sign(a_k X_k) \Big(a_k X_k + \sum_{i\neq k} a_i X_i \cdot 1_{X_k \neq 0}  \Big) \Big] \\ 
&= \E \Big[|a_k X_k| + \sign(a_k X_k) \sum_{i\neq k} a_i X_i \cdot 1_{X_k \neq 0}  \Big] .
\end{align*}
    Since $X_i$ and $X_k$ are mean-zero and pairwise independent for $i\neq k$, we have  $\BE[X_if(X_k)]=\BE[X_i] \cdot \BE[f(X_k)]=0$ for any function $f$. 
    Therefore,  
\begin{align*}
\E \Big[\Big|\sum_i a_i X_i\Big| \cdot 1_{X_k \neq 0} \Big] & ~~\geq~~  \E[|a_k X_k|] + \sum_{i\neq k} \E \Big[ \sign(a_k X_k) \cdot a_i X_i \cdot 1_{X_k \neq 0}  \Big] ~~=~~ \E[|a_k X_k|]. \qedhere
\end{align*} 
\end{proof}

\begin{proof}[Proof of Theorem~\ref{lem:pairwiseIndep}]
Combining Claim~\ref{claim:pairwiseIndepHeart} with Claim~\ref{claim:CombiningYks} completes the proof of Theorem~\ref{lem:pairwiseIndep}.
\end{proof}

%\newpage
\section{Online Discrepancy under Uncorrelated Arrivals}
\label{sec:discuncor}
\newcommand{\CT}{\mathcal{T}}
\newcommand{\CE}{\mathcal{E}}
\newcommand{\be}{\mathbf{e}}
\newcommand{\boldf}{\mathbf{f}}

In this section we consider the vector balancing problem in the special case when the input distribution has uncorrelated coordinates. All our upper and lower bounds will then follow from choosing a suitable basis to reduce the original problem to a basis with uncorrelated coordinates. 

\subsection{Upper Bounds}

 We say a vector in $\BR^d$ is \emph{$s$-sparse} if it has at most $s$ non-zero coordinates. The following lemma bounds the discrepancy for uncorrelated sparse distributions.

\mainGeneralLemma*

%--------------------------------------
\begin{proof}[Proof of Lemma~\ref{lemma:main}]
Our algorithm will use the same potential function approach described in \S\ref{sec:overview}, and uses our anti-concentration lemma from \S\ref{sec:anticonc} to argue that the potential  always remains polynomially bounded.

\paragraph{Algorithm.}
 At any time step $t$, let $d_{t} = \chi_1 v_1 + \ldots + \chi_t v_t$ denote the current discrepancy vector after the signs $\chi_1, \ldots, \chi_t \in \{\pm1\}$ have been chosen. Set $\lambda = \frac1{2s}$ and define the potential function 
\[ \Phi_{t} := \sum_{i \in [n]} \cosh(\lambda d_{t}(i)).\]
When the vector $v_t$ arrives, the algorithm chooses the sign $\chi_t$ that minimizes the increase $\Phi_t - \Phi_{t-1}$.

\paragraph{Bounded Positive Drift.} Let us fix a time $t$. To simplify the notation, let $\Delta\Phi = \Phi_t - \Phi_{t-1}$, let $d = d_{t-1}$, and let $v = v_t$. 

After choosing the sign $\chi_t$, the discrepancy vector $d_t = d + \chi_t v$. To bound the change $\Delta \Phi$, since  $\cosh'(x) = \sinh(x)$ and $\sinh'(x) = \cosh(x)$, using Taylor expansion 
\begin{align*}
    \  \Delta{\Phi} &= \sum_i \Big( \lambda\sinh(\lambda d(i)) \cdot(\chi_t v(i))  + \frac{\lambda^2}{2!} \cosh(\lambda d(i)) \cdot (\chi_t v(i))^2 + \frac{\lambda^3}{3!} \sinh(\lambda d(i))\cdot (\chi_t v(i))^3 + \cdots  \Big), \\
    \                & \le \sum_i \Big( \lambda \sinh(\lambda d(i))  \cdot(\chi_t v(i)) + \lambda^2 \cosh(\lambda d(i))  \cdot(\chi_t v(i))^2 \Big),
\end{align*}
where the last inequality follows since $|\sinh(x)| \le \cosh(x)$ for all $x \in \BR$, and since $|\chi_t v(i)|\le 1$ and $\lambda < 1$,  the higher order terms in the Taylor expansion are dominated by the first and second order terms.

Set $L = \sum_i \sinh(\lambda d(i))v(i)$, and $Q^* =  \sum_i \cosh(\lambda d(i))v(i)^2$, and $Q = \sum_i |\sinh(\lambda d(i))|v(i)^2$. Since $\cosh(x) \le |\sinh(x)|+1$ for $x \in \BR$ and $|v(i)|\le 1$, we have $Q^* \le Q + n$. Therefore,
\[ \Delta \Phi \le \chi_t \cdot \lambda \cdot L + \lambda^2 \cdot Q+\lambda^2 n.\]
Since, the algorithm chooses $\chi_t$ to minimize the increase in the potential:
\[ \Delta \Phi \le - \lambda \cdot |L| + \lambda^2 \cdot Q + \lambda^2n.\]
Now, since $\BE_{v}[v(i)v(j)] = 0$ for all $i,j \in [n]$, we can apply Theorem \ref{lemma:anticonc} with $X_i = v(i)$ and $a_i=\sinh(\lambda d(i))$ to get that $\BE_{v}[|L|] \ge \frac{1}{s} \cdot \BE[Q] = 2\lambda \cdot \BE[Q]$, which yields that 
\[ \BE_{v}[\Delta \Phi] ~~\le~~ - \lambda \cdot \BE_v[|L|] + \lambda^2 \cdot \BE_v[Q] + \lambda^2n ~~\le~~ -\lambda^2 \cdot \BE_v[Q] + \lambda^2 n ~~\le~~ n.\]

\paragraph{Discrepancy Bound.} 
The above implies that for any time $t \in [T]$, the expectation $\BE[\Phi_t] \le nT$. By Markov's inequality and a union bound over the $T$ time steps, with probability at least $1 - T^{-2}$, the potential $\Phi_t \le nT^4$ for every time $t \in [T]$. Since at any time $t$, we have $\cosh({\lambda \infnorm{d_t}}) \le \Phi_t$, this implies that with probability at least $1-T^{-2}$, the discrepancy at every time is  
\[ O\left(\frac{\log(nT^4)}{\lambda}\right) = O(s(\log n + \log T)),\]
which finishes the proof of  Lemma~\ref{lemma:main}.
\end{proof}

%\subsection{Online Discrepancy under General Correlated Distributions}

%Next, we use this lemma to prove Theorem~\ref{thm:signedseries}. 

\subsection{Lower Bounds} \label{sec:lowerBounds}
We now show that the dependence on $s$ and $\log T$ in Lemma~\ref{lemma:main}, cannot be improved up to polynomial factors.
In particular,
a lower bound of $\Omega(s^{1/2}$), even when the time horizon is $T=n$, follows directly from the following more general statement for the vector balancing problem under distributions with uncorrelated coordinates. This general version will later also imply our lower bounds for geometric discrepancy.
%Below, we only consider the case when the time horizon $T=n$.

\begin{lemma}\label{lemma:main_lb}
    Let $\boldp$ be a distribution supported over vectors in $[-1,1]^n$ with $\ell_2$-norm $k$, such that for every $i\neq j \in [n]$ we have $\BE_{v \sim \boldp}[v(i)v(j)]=0$. Then, for any online algorithm that receives as input vectors $v_1, \ldots, v_n$ sampled i.i.d. from $\boldp$, with probability at least $3/4$, the discrepancy is $\Omega(k)$ at some time $t\in [n]$.   
\end{lemma}

We remark that the above lower bound may not hold if the algorithms are offline.

\begin{proof}[Proof of Lemma~\ref{lemma:main_lb}]
    Since the distribution $\boldp$ over inputs is fixed, we may assume that the algorithm is deterministic. Let $d_{t} = \chi_1v_1 + \ldots + \chi_{t}v_{t}$ denote the discrepancy vector at any time $t \in [n]$. Consider the quadratic potential function: 
    \[ \Phi_t ~~:=~~ \|d_t\|^2_2 ~~=~~ \sum_{i \in [n]} |d_t(i)|^2. \]

%The main intuition behind the proof is that if $\infnorm{d_t}$ is small, then the potential $\Phi_t$ must increase significantly. If this happens for a long period, then the discrepancy cannot be small.

%\color{red}    
We will need the following claim that shows $\Phi_t$ increases in expectation for any online algorithm. Let us define $\Delta \Phi_t = \Phi_t - \Phi_{t-1}$. 

\begin{claim}
\label{claim:exp}
Conditioned on any $v_1, \ldots, v_{t-1}$ and signs $\chi_1, \ldots, \chi_{t-1}$ such that  $\infnorm{d_{t-1}} \le \frac{k}{4}$, we have
\begin{align}
    \BE_{v_t}[\Delta\Phi_t] &\ge k^2/2
\end{align}
where the expectation is taken only over the update $v_t \sim \boldp$.
\end{claim}

\begin{proof}
    Set $\Delta \Phi = \Delta\Phi_t$, vector $v = v_t$, and $d=d_{t-1}$. When the update $v$ arrives, note that $d_t = d + \chi_t v$. Therefore, the increase in the potential is given by
\begin{align}
    \label{eqn:deltaphi}
    \ \Delta\Phi ~~=~~ \sum_{i=1}^n \Big(2 d(i)\cdot  \chi_tv(i) +  \big(\chi_t v(i) \big)^2\Big) ~~=~~ 2\chi_t \Big(\sum_{i=1}^n d(i) v(i)\Big) + \|v\|_2^2 ~~=~~ 2 L + k^2,
\end{align}
    where $L = \chi_t \left(\sum_{i=1}^n  d(i) v(i)\right)$. 
    
    To bound the expected value of $L$, we use Jensen's inequality and $\BE_v[v(i)v(j)]=0$ for $i\neq j$ to get:
    \begin{align*}
        \ (\BE_v[L])^2 &~~\le~~ \BE_v[L^2] ~~=~~ \sum_{i=1}^n |d(i)|^2 \cdot  \BE_v[v(i)^2] + \sum_{i\neq j} d(i)d(j)\cdot \BE_v[v(i)v(j)]\\
        \ &~~=~~ \sum_{i=1}^n |d(i)|^2 \cdot  \BE_v[v(i)^2] ~~\le~~ \infnorm{d}^2 \cdot \sum_{i=1}^n \BE_v[v(i)^2] ~~=~~ \infnorm{d}^2 k^2 ~~\le~~ \frac{k^4}{16}.
\end{align*}
    Therefore, plugging the above in \eqref{eqn:deltaphi}, we get
    \[ \BE_v[\Delta\Phi] ~~\ge~~ -2\cdot|\BE_v[L]| + k^2 ~~\ge~~ -2 \cdot \Big( \frac{k^4}{16} \Big)^{1/2} + k^2 ~~\ge~~ \frac{k^2}2. \qedhere \]
\end{proof}

To prove Lemma~\ref{lemma:main_lb} using the last claim, we define $\tau$ to be the first time that $\infnorm{d_\tau} > {k}/{4}$ if such a $\tau$ exists, or $\tau = n$ otherwise. Let us define a new potential $\Phi^*_t$ which remains the same as $\Phi_t$ for $t \le \tau$ and increases by $k^2/2$ deterministically for every $t > \tau$. 

Note that for all possible random choices,  
\[ \Phi^*_n ~~\le~~ \Phi_{\tau-1} + \frac{nk^2}{2} ~~\le~~ \frac{nk^2}{16} + \frac{nk^2}{2},\]
where the second inequality holds since $\infnorm{d_{\tau-1}} \le k/4$ and therefore, $\Phi_{\tau-1} \le \frac{1}{16} \cdot nk^2$.

Moreover, let $\CE$ be the event that $\infnorm{d_t} \le k/4$ for every $t \le n$. Note that when $\CE$ occurs then the final potential $\Phi_n^* \le \frac{1}{16}\cdot{nk^2}$. Defining $p = \BP[\CE]$, we have 
 \begin{equation}
    \label{eqn:upperbound}
     \BE[\Phi^*_n] ~\le~ p \cdot \frac{nk^2}{16} + (1-p)\left(\frac{nk^2}{16} + \frac{nk^2}{2}\right) ~=~ \frac{nk^2}{16} + (1-p)\frac{nk^2}{2}. 
 \end{equation}
 
Moreover, from Claim \ref{claim:exp} and the definition of $\Phi^*_t$, it follows that $\BE[\Phi^*_n] \ge \frac12\cdot {nk^2}$. Comparing this with \eqref{eqn:upperbound} yields that $p \le 1/8$. Hence, with probability at least $7/8$, the discrepancy must be $k/4$ at some point. 
\end{proof}

\paragraph{Dependence on $T$.}

We next show that the discrepancy must be  $\Omega(({\log T}/\log \log T)^{1/2})$ with high probability even when $n=O(1)$ (we assume $n \geq 2$ throughout this discussion). We only sketch the proof here as the arguments are standard. The idea is that for large $T$, there is a high probability of getting a long enough run of consecutive vectors with each $v_t$ almost orthogonal to $d_{t-1}$. 

%The idea is simply that for large $T$, there is a high probability of getting a long enough run of consecutive vectors with each $v_t$ almost orthogonal to $d_{t-1}$. 

Let $\boldp$ be the uniform distribution\footnote{Our argument works for a wide class of distributions $\boldp$, as long as for any $d_{t-1} \in \mathbb{R}^n$,
the random incoming vector $v_t$ sampled from $\boldp$ has a non-trivial probability of having a small inner product with $d_{t-1}$.
We only give the argument for the uniform distribution on the unit sphere for simplicity.} over vectors on the unit sphere $S^{n-1}$.
For any vector $u \in \R^n$, and $v$ sampled from $\boldp$, there is a universal constant $c$ so that for all $\delta \leq 1$, 
we have  $\BP[ |\ip{u}{v}| \leq \delta \|u\|_2 /n^{1/2}]  \geq c \delta.$
 
%So for any vector $d$,
%\[ \Pr[d \cdot v \leq \frac{\epsilon \|d\|_2}{n^{1/2}}] \geq  c \epsilon.\]

Let $\beta \geq 1 $ be some parameter that we optimize later. 
Setting $\delta = 1/(4\beta)$ gives that whenever $\|d_{t-1}\|_2 \leq \beta n^{1/2}$,
there is at least $c/(4 \beta)$ probability that $| \ip{d_{t-1}}{v_t} | \leq 1/4$, and hence irrespective of the sign $\chi_t$,
\[ \|d_t\|_2^2 ~~ \geq ~~ \|d_{t-1}\|_2^2 - 2 |\ip{d_{t-1}}{v_t} | + \|v_t\|^2_2 ~~\geq~~ \|d_{t-1}\|_2^2 + 1/2. \]

So for any $\tau$ consecutive steps, with at least $(c/4\beta)^{\tau}$ probability, this happens at every step {(or the $\ell_2$-discrepancy already exceeds $\beta n^{1/2}$ at some step)}, and hence the discrepancy has $\ell_2$-norm at least $\Omega(\tau^{1/2})$.

Partitioning the time horizon $T$ into ${T}/{\tau}$ disjoint blocks, and setting $\beta = \log (T)$,
and $\tau = \Omega(\log T/\log \log T)$, the probability 
such a run does not occur in any block is at most  
$(1 - (c/4\beta)^{\tau})^{(T/\tau)} = T^{-\Omega(1)}$  by our choice of the parameters.
This gives the claimed lower bound.

\section{Online Vector Balancing: Polynomial Bounds}
\label{sec:signedseries}

%\subsubsection{Proof of Theorem \ref{thm:signedseries} by working in the Eigen Basis}
In this section, we prove our vector balancing result for arbitrary distributions. % We restate the theorem for convenience.
\mainVector*

%\subsubsection*{Proof of Theorem \ref{thm:signedseries} by working in the Eigenbasis}

\begin{proof}[Proof of Theorem \ref{thm:signedseries}]
    Without loss of generality, we may assume that the distribution $\boldp$ is \emph{symmetric}, i.e. both $v$ and $-v$ have the same probability density, since we can always multiply the incoming vector $v$ with a Rademacher $\pm 1$ random variable  without changing the problem. Let $P \in \BR^{d \times d}$ denote the covariance matrix of our input distribution, and since $\boldp$ is symmetric, we get $P = \BE_{v \sim \boldp}[vv^T]$. Let $U$ denote the orthogonal matrix whose columns $u_1, \ldots, u_n$ form an eigenbasis for $P$. Note that in terms of its spectral decomposition, $P = \sum_{k=1}^n \lambda_k u_ku_k^T$ for $\lambda_k \in \R$.

 To prove our discrepancy bound, instead of working in the original basis, we will view our problem as a vector balancing problem in the basis given by the columns of $U$. Now the update sequence is given by $w_1,\ldots,w_T$ where $w_t = \frac{1}{\sqrt{n}} \cdot U^Tv$ is the normalized update vector in the basis $U$. 

   % we will view our problem as a vector balancing problem over a fixed orthonormal eigenbasis for $P$ and apply the algorithm given by Lemma \ref{lemma:main}. 
    %

Since $\|v\|_2 \le \sqrt{n}$ and orthogonal matrices preserve $\ell_2$-norm, we have $\|U^Tv\|_2 = \|v\|_2 \le \sqrt{n}$. It follows that for any $t$, we have $\infnorm{w_t} \le \|{w_t}\|_2 = \frac{1}{\sqrt{n}} \cdot \|{U^Tv}\|_2 \le 1$. Furthermore, any two coordinates of the update vectors $w_t$'s are uncorrelated, i.e., for any $i\neq j\in[n]$ we have
    \[ \BE[w_t(i) \cdot w_t(j)] ~~=~~ \frac{1}{n} \BE[\ip{u_i}{v}\ip{u_j}{v}] ~~=~~ \frac{1}{n}  \BE[u_i^Tv v^Tu_j] ~~=~~ \frac{1}{n}  u_i^TPu_j ~~=~~ 0,\]
where the last equality holds since $P=\sum_{k=1}^n \lambda_k u_k^Tu_k$.

    Thus, we can use the online algorithm from Lemma \ref{lemma:main} to select signs $\chi_1,\ldots,\chi_T \in \{\pm 1\}$. Let $d_t = \chi_1 v_1 + \ldots + \chi_t v_t$ denote the discrepancy in the original basis. Now using the trivial bound of $s \le n$ on  sparsity in Lemma~\ref{lemma:main}, we get  that with high probability, 
    \[\frac{1}{\sqrt{n}} \infnorm{U^Td_t} ~=~ O(n(\log n + \log T)).\]
    Again, using that orthogonal matrices preserve $\ell_2$-norm, 
    \[ \infnorm{d_t} ~~\le~~ \|d_t\|_2 ~~=~~ \|U^Td_t\|_2 ~~\le~~ \sqrt{n} \cdot \infnorm {U^Td_t}~~ = ~~ O(n^{2}(\log n + \log T)). \qedhere\]
\end{proof}

%%%%%%%%%%%%%%%%%%%%%%%%%%%%%%%%%%%%%%%%%%

\section{Online Geometric Discrepancy: Polylogarithmic Bounds}
\label{sec:geomdisc}

In this section, we will prove our results on geometric discrepancy problems. For this, we will need a special basis of orthogonal functions on the unit interval called the Haar system. We briefly review its properties.

\subsection{Preliminaries: Haar System}
\label{sec:haar}

 Let $\Psi:\BR \to \BR$ denote the \emph{mother wavelet} function
%$\Psi = \ind_{[0,1/2)} - \ind_{[1/2,1)}$.
\begin{equation*}
        \Psi(x) = \begin{cases}
                1 &\text{if } 0 \le x < \frac12 \\
                -1 &\text{if } \frac12 \le x < 1 \\
                0 &\text{otherwise.}
        \end{cases}
\end{equation*}

The \emph{unnormalized} Haar wavelet functions (recall Figure \ref{fig:haar1d}) are defined as follows: 
let $\Psi_{0,0}(x)=1$ for all $x \in \BR$
, and for any $j \in \RN^*$ and $0\le k < 2^{j-1}$ define 
\[ \Psi_{j,k}(x) := \Psi(2^{j-1}x - k).\] 
%Note that this is the same definition as in~\eqref{eq:defnPsi}, 
%for any $j > 0$ the function $\Psi_{j,k}$ takes the value $1$ on the first half of the interval $I_{j-1,k} = \big[ k2^{-{(j-1)}}, (k+1)2^{-{(j-1)}} \big)$, and $-1$ on the second half, and is 0 otherwise. 

We call $j$ as the \emph{scale} and $k$ as the \emph{shift} of the wavelet.%in analogy with the classical Fourier basis.

%Let me define $\Psi_{(-1,0)} = \ind$ to be the constant function $1$.

%One can also define a dyadic martingale called the Haar martingale based on these functions which is what we were doing before. 
The Haar wavelet functions have nice orthogonality properties. In particular, let $x$ be drawn uniformly from the unit interval $[0,1]$. Then, one can easily check that  
\begin{equation}
\begin{aligned}
\label{eqn:haarprop}
\ &\BE_{x} [\Psi_{j,k}(x)^2] = 2^{-(j-1)} &\text{ for } j > 0,\\
\ &\BE_{x} [\Psi_{j,k}(x)] = 0 &\text{ for } j > 0, \\
\ &\BE_{x}[ \Psi_{j,k}(x)\Psi_{j',k'}(x)] = 0 &\text{ unless }j=j' \text{ and } k=k'. 
\end{aligned}
\end{equation}

%but together with the constant function $\ind$

The Haar wavelet functions are not just orthogonal, but they form an orthogonal basis (not orthonormal), called the \emph{Haar system}, for the class of functions on the unit interval with bounded $L_2$-norm. In particular, we have the following proposition where for  $j \in \RZ_{\ge 0}$ we denote  $\CH_j = \bigcup_{0 \le k < 2^{j-1}} \{\Psi_{j,k}\}$ and let $\CH = \bigcup_{j \ge 0} \CH_j$.     % with respect to the inner product $\ip{f}{g} = \BE_{x \in u}[f(x)g(x)]. %The class of functions $\CH_1 = \{\Psi_{j,k}\}_{j,k} \cup {1}$ form an orthogonal basis for the class of continuous functions with bounded $\ell^2$ norm on $[0,1]$. 
\begin{proposition}[\cite{W04}, Chapter 5] 
\label{prop:ortho}
%\begin{itemize}
    %    \item[Norm:] $\BE_{x \sim u} [\Psi_{j,k}(x)^2] = 2^{-j}$.
    %    \item[Mean:] $\BE_{x \sim u} [\Psi_{j,k}] = 0.$
    %    \item[Orthogonality:] $\BE_{x \sim u}[ \Psi_{j,k}\Psi_{j',k'}] = 0$ unless $j=j'$ and $k=k'$.
%\item[Conditioned orthogonality:] Let $I_{j,k} \subseteq I_{j',k'}$, then $\int_{I_{j,k}} \Psi_{j,k} \Psi_{j',k'} = 0$ since $\Psi_{j',k'}$ is constant on the interval $I_{j,k}$.
      %  \item[Basis:] 
    For any $f : [0,1] \to \BR$ such that $\BE_{x}[f(x)^2] < \infty$, we have  
        \[ f = \sum_{h \in \CH} \fhat(h) \cdot h(x) \]
        where $\fhat(h) = \frac{\BE_{x}[f(x)h(x)]}{\BE_x[h(x)^2]}$ is the corresponding coefficient in the Haar system basis for $h \in \CH$. 
%\end{itemize}
\end{proposition}

Indeed, since the Haar system forms an orthogonal basis, we also have that
\[ \BE_x[f(x)^2] = \sum_{h \in \CH} \fhat(h)^2 \cdot \BE_x[h(x)^2].\]
%but we will not need the above fact for our purposes.
%\nbnote{Agree with above}
A simple corollary of Proposition \ref{prop:ortho} is that $\CH^{\otimes d}$ is an orthogonal basis for the linear space spanned by all functions over the unit cube $[0,1]^d$ that have a product structure and bounded $L_2$-norm. In particular, let $\bh = (h_1, \ldots, h_d)$ be an element of $\CH^{\otimes d}$ which we will view as a function from $[0,1]^d \to \BR$ by defining $\bh(x) = \prod_{i=1}^d h_i(x(i))$ for $x \in [0,1]^d$. Note that distinct $\bh$ and $\bh'$ are orthogonal since for $x$ drawn uniformly from $[0,1]^d$, 
\begin{equation}
    \label{eqn:haartensorprop}
    \ \BE_{x}[\bh(x)\bh'(x)] ~~=~~ \prod_{i=1}^d \BE_{x(i)}[h_i(x(i))h_i'(x(i))] ~~=~~ 0.
\end{equation}
Moreover, any product function can be expressed by functions in $\CH^{\otimes d}$ as given in the following proposition\footnote{More generally, Proposition~\ref{prop:tensorHaarBasis} holds for any $L_2$-integrable function $f \in L_2([0,1]^d)$, as the linear span of product functions with domain $[0,1]^d$ is dense in $L_2([0,1]^d)$.}.

\begin{proposition} \label{prop:tensorHaarBasis}   %\begin{itemize}
    %    \item[Norm:] $\BE_{x \sim u} [\Psi_{j,k}(x)^2] = 2^{-j}$.
    %    \item[Mean:] $\BE_{x \sim u} [\Psi_{j,k}] = 0.$
    %    \item[Orthogonality:] $\BE_{x \sim u}[ \Psi_{j,k}\Psi_{j',k'}] = 0$ unless $j=j'$ and $k=k'$.
%\item[Conditioned orthogonality:] Let $I_{j,k} \subseteq I_{j',k'}$, then $\int_{I_{j,k}} \Psi_{j,k} \Psi_{j',k'} = 0$ since $\Psi_{j',k'}$ is constant on the interval $I_{j,k}$.
      %  \item[Basis:] 
    For any $f : [0,1]^d \to \BR$ such that $f(x) = \prod_{i=1}^df_i(x(i))$ for some $f_i : [0,1] \to \BR$ satisfying $\BE_{x(i)}[f_i(x(i))^2] < \infty$, we have that 
        \[ f = \sum_{\bh \in \CH^{\otimes d}} \fhat(\bh) \bh, \]
        where $\fhat(\bh) =  \frac{\BE_{x}[f(x) \bh(x)]}{\BE_x[ \bh(x)^2]}$. %Furthermore, $\fhat(\bh) = \prod_{i=1}^d \fhat_i(h_i)$. 
%\end{itemize}
%\nbnote{Should we not have  a similar normalization here?}
\end{proposition}
\begin{proof}
   Expressing each $f_i$ in the Haar system basis using Proposition \ref{prop:ortho}, we get the statement of the proposition by tensoring.
\end{proof}

%For a function $f: [0,1] \to \BR$, we will use the notation $\fhat(h)$ to denote its coefficient in the Haar system basis for $h\in \CH$. 

Let $\CH_{\le j} = \bigcup_{j' \le j} \CH_{j'}$, and define $\CH_{< j}, \CH_{>j}, \CH_{\ge j}$ analogously. Then, we have the following lemma about the Haar system decomposition of indicator functions of dyadic intervals. 
\begin{proposition}
    \label{prop:indhaar}
    Let $\ind_{I_{\ell,m}}$ denote the indicator function for the interval $I_{\ell,m} = \big[m2^{-\ell}, (m+1)2^{-\ell} \big)$. Then, 
    \begin{align*}
        \  \sum_{h \in \CH_0}|\widehat{\ind}_{I_{\ell,m}}(h)| &=2^{-\ell} , \\
        \  \sum_{h \in \CH_j}|\widehat{\ind}_{I_{\ell,m}}(h)| &=2^{-(\ell + 1 - j)} \text{ for any } 1 \le j \le \ell \text{ and}\\
         \  \widehat{\ind}_{I_{\ell,m}}(h) &= 0 \text{ for any } h \in \CH_{> \ell}.
    \end{align*}
In particular, we  have $\sum_{h \in \CH} |\widehat{\ind}_{I_{\ell,m}}(h)| = \sum_{h \in \CH_{\le \ell}} |\widehat{\ind}_{I_{\ell,m}}(h)| = 1.$
\end{proposition}
\begin{proof}
    First, observe that for any $j > \ell$, either $\Psi_{j,k}(x)=0$ identically on the interval $I_{\ell,m}$ or it takes $+1$ and $-1$ values on equal size sub-intervals of $I_{\ell,m}$, so that $\BE_x[1_{I_{\ell,m}}(x)\Psi_{j,k}(x)]=0$.
    
    For $\Psi_{0,0}$, notice that $\BE_x[1_{I_{\ell,m}}(x)\Psi_{0,0}(x)] = 2^{-\ell}$ and $\BE_x[\Psi_{0,0}(x)^2] = 1$. Therefore, we have
    \[ \sum_{h \in \CH_0} |\widehat{\ind}_{I_{\ell,m}}(h)| =2^{-\ell}.\]

    Now consider any $1 \le j \le \ell$. Then, there exists a unique $0\le k^* < 2^{j-1}$ such that $\Psi_{j,k^*}$ takes the constant value $+1$ or $-1$ identically on the interval $I_{\ell,m}$, and the function $\Psi_{j,k}$ is identically zero on the interval $I_{\ell,m}$ for any $k \neq k^*$. It follows that  $\BE_x[1_{I_{\ell,m}}(x)\Psi_{j,k^*}(x)] = \pm 2^{-\ell}$, $\BE_x[\Psi_{j,k^*}(x)^2] = 2^{-(j-1)}$ and  $\BE_x[1_{I_{\ell,m}}(x)\Psi_{j,k}(x)]=0$ for any $k \neq k^*$. Therefore, for $1 \le j \le \ell$, we have
    \[ \sum_{h \in \CH_j} |\widehat{\ind}_{I_{\ell,m}}(h)| =2^{-(\ell + 1 - j)}.\]

    From the above, it also follows that \[\sum_{h \in \CH} |\widehat{\ind}_{I_{\ell,m}}(h)| ~~ = ~~ \sum_{h \in \CH_{\le \ell}} |\widehat{\ind}_{I_{\ell,m}}(h)| ~~ = ~~ 2^{-\ell} + \sum_{j=1}^{\ell} 2^{-(\ell + 1 - j)} ~~=~~ 2^{-\ell} + (1-2^{-\ell}) ~~=~~ 1. \qedhere\]
\end{proof}

We also get a similar proposition about dyadic boxes. In particular, let $\bl = (\ell_1, \ldots, \ell_d)$ for non-negative integers $\ell_i$'s and let $\bm = (m_1, \ldots, m_d)$ for integers $0 \le m_i < 2^{\ell_i}$. Let $\CH^{\otimes d}_{\le \bl} = \CH_{\le \ell_1} \times \cdots \times \CH_{\le \ell_d}$. Then, for the dyadic box 
\[ I_{\bl,\bm} = I_{\ell_1,m_1} \times \cdots \times I_{\ell_d,m_d},\]
we have the following proposition. Below we write $\min\{\be, \boldf\}$ to denote the vector whose $i^{\text{th}}$ coordinate is $\min\{\be(i),\boldf(i)\}$ for $\be, \boldf \in \BR^d$.
\begin{proposition}
    \label{prop:indhaarbox}
    Let $\ind_{I_{\bl,\bm}}$ denote the indicator function for the dyadic box $I_{\bl,\bm}$. Then, 
    \begin{align*}
        \ \sum_{\bh \in  \CH^{\otimes d}_{\bj} }|\widehat{\ind}_{I_{\bl,\bm}}(\bh)| &=2^{-\|\min\{\bl, \bl + \mathbf{1} - \bj\}\|_1} \text{ for any } \bj \le \bl \text{ and} \\
         \ \widehat{\ind}_{I_{\bl,\bm}}(\bh) &= 0 \text{ for any }\bh \notin \CH_{\le \bl}.
    \end{align*}
    In particular, we have $\displaystyle\sum_{\bh \in \CH^{\otimes d}} |\widehat{\ind}_{I_{\bl,\bm}}(\bh)| = \sum_{\bh \in  \CH_{\le \bl}^{\otimes d}}|\widehat{\ind}_{I_{\bl,\bm}}(\bh)| = 1.$
\end{proposition}

%\nbnote{In the first equality, under the summation, I think $\CH^{\otimes d}_{\le bj}$  should just be $\CH^{\otimes}_{bj}$, i.e. without the $\le$. ?}
 
%\nbnote{The rest looks ok.}
%\begin{proposition}
    %\label{prop:indhaarbox}
    %Let $\ind_{I_{\bl,\bm}}$ denote the indicator function for the dyadic box $I_{\bl,\bm}$. Then, 
    %\begin{align*}
      %  \ \sum_{\bh \in  \CH^{\otimes d}_{\le \bj} }|\widehat{\ind}_{I_{\bl,\bm}}(\bh)| &=2^{-\|\bl\|_1} \text{ for any } \bj \le \bl \text{ and}\\
     %    \ \widehat{\ind}_{I_{\bl,\bm}}(\bh) &= 0 \text{ for any }\bh \notin \CH_{\le \ell}.
    %\end{align*}
    %In particular, we have %$\displaystyle\sum_{\bh \in \CH^{\otimes d}} |\widehat{\ind}_{I_{\bl,\bm}}(\bh)| = \sum_{\bh \in  \CH_{\le \bl}^{\otimes d}}|\widehat{\ind}_{I_{\bl,\bm}}(\bh)| \le 1.$
%\end{proposition}
The proof of the above proposition follows from Proposition \ref{prop:indhaar} by tensoring.

%In $d$-dimensions, one can define $\CH_d = (\CH_1)^d$ and it seems to me that they should form an orthogonal basis for $L^2$ bounded functions on $[0,1]^d$ but one needs to check this more carefully (Reference?). For our purposes, we only need discrete versions of the above. Let me write $\Psijk{j,k}$ for the $d$-dimensional functions where $\bj,\bk$ are $d$-tuples now. Furthermore, the inner product between a function $f$ on $[0,1]^d$ is defined to be
%\[ \ip{f}{g} = \int_{[0,1]^d} f(\bx)g(\bx) d\bx = \sum_{\bj,\bk}  \fhat_{\bj,\bk} \hat{g}_{\bj,\bk}\|\Psijk{j,k}\|_2^2 .\]

%--------------------------------------
\subsection{Online Interval Discrepancy Problem}
Now we prove Theorem~\ref{thm:interval} for the $d$-dimensional interval discrepancy problem.
Let $\boldx = (x_1, \ldots, x_T)$ be a sequence of points in $[0,1]^d$ and let $\chi \in \{\pm 1\}^T$ be a signing. For any interval $I \subseteq [0,1]$ and time $t \in [T]$, recall that the discrepancy of interval $I$ along coordinate direction $i$ at time $t$ is denoted
\[ \disc_t^i(I,\boldx,\chi) := \Big|\chi_1 \ind_I(x_1 (i)) + \cdots + \chi_t \ind_I(x_t (i)) \Big|. \]
We will just write $\disc_t^i(I)$ when the input sequence and signing is clear from the context.

\subsubsection{Upper Bounds}

To maintain the discrepancy of all intervals, it will suffice to bound the discrepancy of every dyadic interval $I_{j,k}=[k2^{-j}, (k+1)2^{-{j}})$ of length at least $1/T$ along every coordinate direction $i$. Let $\CD = \{ I_{j,k}  \mid  0\le j \le \log T, 0 \le k <2^j\}$. Then, we prove the following.

%We first give an online algorithm that maintains a small discrepancy 
%We first give an online algorithm that maintains discrepancy $O(d \log^2 T)$ with high probability for every dyadic interval along any coordinate direction $i \in [d]$, which gives us the final discrepancy bound,  
%\[ \max_{i \in [d]} \disc^i_t(I) \le O(d\log^3 T). \]

\begin{lemma} 
    \label{lemma:int_main}
    Given any sequence $x_1, \ldots, x_T$ sampled independently and uniformly from $[0,1]^d$, there is an online algorithm that chooses a signing such that w.h.p. for every time $t \in [T]$, we have
     \[  \max_{i \in [d]} \disc^i_t(I) = O(d\log^2 T) ~~\text{ for all } I \in \CD.\]  
\end{lemma}

%\begin{proposition}
%    \label{prop:int_truncation}
%    Let $x_1,\ldots, x_T \sim_u [0,1]^d$ and let $\chi \in \{\pm1\}^T$ be such that $\max_{i \in [d]} \disc^i_t(I_{j,k}) \le D$ for every dyadic interval $I_{j,k}$ where $0 \le j \le \log T$ and $0 \le k < 2^{j}$ and every time $t \in [T]$. Then, with high probability over the input sequence, $\max_{i \in [d]} \disc^i_t(I) \le O((D + d)\log T)$ for every interval $I$ at every time $t$.
%\end{proposition}

Before proving Lemma~\ref{lemma:int_main}, we first show why it implies the upper bound in Theorem~\ref{thm:interval}.

\begin{proof}[Proof of the upper bound in~Theorem \ref{thm:interval}]

%To bound the discrepancy of an arbitrary interval $I$ along a coordinate axis $i \in [d]$, it suffices to get a bound on the discrepancy of every interval of the form $[\frac{a}{T}, \frac{b}{T}]$ for $a, b \in \{0,\ldots, T\}$ along that axis.
    Without loss of generality, it suffices to consider half-open intervals. Every half-open interval $I \subseteq [0,1]$ can be decomposed as a union of at most $2\log T$ disjoint dyadic intervals in $\CD$ and two intervals $I_1 \subseteq I_{\log T,k}$ and $I_2 \subseteq I_{\log T,k'}$ for some $0 \le k,k' < T$. Note that the length of $I_1$ and $I_2$ is at most $2^{-\log T} = 1/T$. We can then write, 
    \[ \disc^i_t(I) \le (2\log T) \cdot \max_{I \in \CD} \disc^i_t(I) + \disc^i_t(I_1) + \disc^i_t(I_2).\]

    Applying the algorithm from Lemma~\ref{lemma:int_main}, the discrepancy of every dyadic interval can be bounded w.h.p. by $O(d\log^2T)$. The last two terms can be bounded by $N_1$ and $N_2$ respectively where $N_1$ (resp. $N_2$) is the number of points whose projections on any of the $i$ coordinates is in $I_1$ (resp. $I_2$). % $z \in [0,1]^d$ such that $z(i) \in I_1$ (resp. $x_t(i) \in I_2$) in the whole sequence $x_1, \ldots, x_T$.  

   % contained inside some dyadic intervals of size $1/T$.
  %  Since, the length of intervals $I_1$ and $I_2$ is at most $1/T$,
    The probability that a random point $z$ drawn uniformly from $[0,1]^d$ has some coordinate $z(i)$  for $i \in [d]$ in $I_1$ or $I_2$ is at most $2d/T$. It follows that $\BE[N_1+N_2] \le 2d$, so by Chernoff bounds, with probability at least $1-T^{-4}$, the number $N_1 + N_2 \le 4d\log T$.
    
 %   the number of points such that $x_t(i) \in I_1 \cup I_2$ for any $i \in [d]$ is at most $4d\log T$. Now, using Lemma \ref{lemma:int_main} the discrepancy of every dyadic interval is bounded by $d\log^2T$ with high probability. 

Overall, w.h.p. for any interval $I$, we have
    \[ \max_{i \in [d]} \disc^i_t(I) ~~\le~~ 2\log T \cdot (d \log^2 T) + 4d\log T ~~=~~ O(d\log^3 T). \qedhere \]
\end{proof}

Next, we prove the missing Lemma~\ref{lemma:int_main}.

\begin{proof}[Proof of Lemma~\ref{lemma:int_main}]
    %Let $\CH' = \CH_1 \cup \cdots \cup \CH_{\log T}$ denote the Haar wavelet functions with scale parameter between one and $\log T$. Note that $\CH'$ does not include the constant Haar function $\Psi_{0,0}$ and $|\CH'|=T-1$. 
    
    We will consider the $d$-dimensional interval discrepancy problem as a vector balancing problem  in $\left | [d]\times \CH_{\le \log T} \right |$ dimensions, where $\CH_{\le \log T}$ are the Haar wavelet functions with scale parameter at most $\log T$. Note that $|\CH_{\le \log T}|=T$, so the update vector in the vector balancing version will be $Td$-dimensional. Let us abbreviate $\CH'=\CH_{\le \log T}$.

   % We will consider the $d$-dimensional interval discrepancy problem as a vector balancing problem in $d(T-1)+1$ dimensions where $d(T-1)$ of these coordinates will correspond to the set $[d]\times \CH'$ and the extra coordinate will be indexed by $(1, \Psi_{0,0})$. %Note that $|\CH'|=T-1$, so the update vector in the vector balancing version will be $(T-1)d)+1$-dimensional. 
  % Note that we limit the frequencies $j$ to be at most $\log T$, so the smallest possible intervals we consider are of length $1/T$.

  %  The dimensions of the update vector will correspond to $i \in [d]$ and $h \in \CH_{\le \log T}$ (note that $|\CH_{\le \log T}|=T$). 
    
    At any time when the point $x_t \in [0,1]$ arrives, then the $(i,h)$ coordinate of the update vector $v_t \in [-1,1]^{d \times \CH'}$ is given by
    \[ v_t(i,h) = h(x_t(i)).\]

    Note that all the coordinates $(i,\Psi_{0,0})$ for $i\in [d]$ will always have the same value where $\Psi_{0,0}$ is constant Haar wavelet. So, to apply the online algorithm given by Lemma \ref{lemma:main} we will only consider the subspace spanned by the coordinates $(i,h)$ where $i \in [d]$ and $h \neq \Psi_{0,0}$ and the extra coordinate $(1,\Psi_{0,0})$.
    
    Let us check first that we satisfy the conditions Lemma \ref{lemma:main}. First, note that the $\infnorm{v_t} \le 1$ and the vector $v_t$ has at most $d\log T + 1$ non-zero coordinates, since for any fixed scale $0\le j\le \log T$ and any point $z \in [0,1]$, all but one of the values $\{h(z)\}_{h \in \CH_j}$ are zero. The last condition to check is that the coordinates of the vector $v_t$ are uncorrelated. This is a consequence of \eqref{eqn:haarprop}, since whenever 
     coordinates $(i,h)$ and $(i',h')$ satisfy $i \neq i'$ or $h\neq h'$, we have 
    \[ \BE_{v_t}[v_t(i,h)\cdot v_t(i',h')] ~~=~~ \BE_{x_t}[h(x_t(i))\cdot h'(x_t(i'))] ~~=~~ 0. \]
    To elaborate more, first note that we cannot have $h=h'=\Psi_{0,0}$ since we are working in the aforementioned subspace. Now, if $i\neq i'$ then the coordinates $x_t(i)$ and $x_t(i)$ are sampled independently from $[0,1]$, and $\BE_z[h(z)]=0$ for $h \neq \Psi_{0,0}$ when $z$ is drawn uniformly from $[0,1]$. Otherwise, for $i=i'$ but $h\neq h'$, it follows from the orthogonality of the Haar system that $\BE_z[h(z)h'(z)]=0$.  %(note that we can not have $h=h'=\Psi_{0,0}$ since we are working in the subspace spanned by only $(1,\Psi_{0,0})$
    
    Next, applying the online algorithm from Lemma \ref{lemma:main}, we select signs $\chi_1,\ldots,\chi_T$ such that we get an $\ell_{\infty}$ bound on the vector $d_t = \sum_{l\le t} \chi_l v_l$. In particular, with high probability we have
    \begin{align*}
         \ |d_t(i,h)| & ~~=~~ \Big|\sum_{l\le t} \chi_l h(x_l(i))\Big|  ~~=~~ O(d\log^2 T) ~~\text{ for any } i \in [d],h \in \CH'.
    \end{align*}
Note that the bound on $|d_t(i, \Psi_{0,0})|$ for $i \neq 1$ follows since $|d_t(i, \Psi_{0,0})|=|d_t(1, \Psi_{0,0})|$.

%    For any $i \in [d]$, let us define $d_t(i, \Psi_{0,0}) = d_t(1, \Psi_{0,0}) = \sum_{l \le t} \chi_l$. Note that the above bounds imply that for all $i \in [d]$ and $h \in \CH_{\le \log T}$ (note that $\CH_{\le \log T} = \{\Psi_{0,0}\} \cup \CH'$), we have $|d_t(i,h)| = O(d\log^2 T)$.

    To finish the proof, we need to bound the discrepancy of every dyadic interval in terms of $\infnorm{d_t}$. Note that for any dyadic interval $I \in \CD$, its coefficients in the Haar system basis $\widehat{\ind}_{I}(h)=0$ for $h \in \CH_{> \log T}$ using Proposition \ref{prop:indhaar}. Now, for any $i \in [d]$ and dyadic interval $I \in \CD$, we can write
    \begin{align*}
        \ \disc^i_t(I) &= \Big|\sum_{l \le t} \chi_l \ind_{I}(x_l(i))\Big| ~~=~~ \Big|\sum_{l \le t} \chi_l \sum_{h \in \CH'} \widehat{\ind}_{I}(h) h(x_l(i))\Big|\\
        \ &= \Big|\sum_{h \in \CH'} \widehat{\ind}_{I}(h) \Big(\sum_{l \le t} \chi_l h(x_l(i))\Big)\Big| ~~=~~ \Big|\sum_{h \in {\CH'}} \widehat{\ind}_{I}(h) d_t(i,h)\Big| \\
        \                &\le  \infnorm{d_t} \cdot \Big(\sum_{h \in \CH'} |\widehat{\ind}_{I}(h)|\Big) ~~\le~~ \infnorm{d_t} ~~=~~ O(d \log^2 T),  
    \end{align*}
    where the second last inequality follows again from Proposition \ref{prop:indhaar}. 
\end{proof}

\subsubsection{Lower Bounds}

\begin{proof}[Proof of the lower bound in Theorem~\ref{thm:interval}]
 Set $A=T/d$. We will again consider the $d$-dimensional interval discrepancy problem as a vector balancing problem in $\left | [d] \times \CH_{\le \log A} \right |$ dimensions where $\CH_{\le \log A}$ are the Haar wavelet functions with scale parameter at most $\log A$. Note that $|\CH_{\le \log T}|=A$, so the update vector in the vector balancing version will be $T$-dimensional. Let us abbreviate $\CH' = \CH_{\le \log A}$.
 
 At any time when the point $x_t \in [0,1]^d$ arrives, then the $(i,h)$ coordinate of the update vector $v_t$ is given by
\[ v_t(i,h) = \begin{cases} 
0 & \text{ if } h = \Psi_{0,0} \\ h(x_t(i)) & \text{ otherwise.} \end{cases} \] 
 
 %and let $\CH' = \CH_1 \cup \cdots \cup \CH_{\log A}$ denote the Haar wavelet functions with scale parameter between one and $\log A$. Note that $\CH'$ does not include the constant Haar function $\Psi_{0,0}$ and $|\CH'|=A-1$. % We will consider the $d$-dimensional interval discrepancy problem as a vector balancing problem in $d(A-1)+1$ dimensions where $d(A-1)$ of these coordinates will correspond to the set $[d]\times \CH'$ and the extra coordinate will be indexed by $(1, \Psi_{0,0})$. %Note that $|\CH'|=T-1$, so the update vector in the vector balancing version will be $(T-1)d)+1$-dimensional. 
  % Note that we limit the frequencies $j$ to be at most $\log T$, so the smallest possible intervals we consider are of length $1/T$.

  %  The dimensions of the update vector will correspond to $i \in [d]$ and $h \in \CH_{\le \log T}$ (note that $|\CH_{\le \log T}|=T$). 
    
    %At any time when the point $x_t \in [0,1]$ arrives, then the $(i,h)$ coordinate of the update vector $v_t$ is given by
%\[ v_t(i,h) = h(x_t(i)).\]

 %   We will apply the online algorithm given by Lemma \ref{lemma:main}. Let us check first that we satisfy the conditions of that lemma. First, note that the $\infnorm{v_t} \le 1$ and the vector $v_t$ has at most $d(\log T + 1)$ non-zero coordinates, since for any fixed scale $0\le j\le \log T$ and any point $z \in [0,1]$, all but one of the values $\{h(z)\}_{h \in \CH_j}$ are zero. The last condition to check is that the coordinates of the vector $v_t$ are uncorrelated. This follows f
    Here we are essentially ignoring the coordinates $(i,h)$ with $h = \Psi_{0,0}$.
    Since for any fixed scale $0< j\le \log A$ and any point $z \in [0,1]$, all but one of the values $\{h(z)\}_{h \in \CH_j}$ are zero, the vector $v_t$ has $d\log A$ non-zero coordinates all of which take value $\pm1$. It follows that the Euclidean norm of any update vector $v_t$ is $\sqrt{d\log A}$. 
    %We shall again identify all the coordinates $(i,h)$ where $h = \Psi_{0,0}$ and consider the subspace where all but one of these coordinates are removed. 
    %that the vector $v_t$ has at most $d(\log A + 1)$ coordinates that are $\pm1$ and rest of the coordinates are zero

Furthermore, from the orthogonality of the Haar system, it follows that the coordinates of the vector $v_t$ are uncorrelated:
    \[ \BE_{v_t}[v_t(i,h)v_t(i',h')] = \BE_{x_t}[h(x_t(i))h'(x_t(i'))] = 0.\]

    Then, applying Lemma \ref{lemma:main_lb}, we get that with probability at least $3/4$, there is a $t \in [T]$ and a coordinate $(i,h)$ with $h \neq \Psi_{0,0}$ such that $|d_t(i,h)| = \Omega(\sqrt{d \log A})$. 

   Let $h = \Psi_{j,k}$ for some $j,k$ where $j > 0$ (recall that coordinates $(i,h)$ where $h = \Psi_{0,0}$ are always $0$). Then, by definition $h = \ind_{I_1} - \ind_{I_2}$ where $I_1$ and $I_2$ are the first and second halves of the interval $I_{j-1,k}$. In this case,  
    \[ |d_t(i,h)| ~~=~~ \Big|\Big(\sum_{s\le t}\chi_t \ind_{I_1}(x_s)\Big) - \Big(\sum_{s\le t}\chi_t \ind_{I_2}(x_s) \Big) \Big| ~~\le~~ 2 \max\Big\{|\disc_t(I_1)|,|\disc_t(I_2)|\Big\}.\]
    Therefore, substituting $A=T/d$, there exists an interval $I$ such that $\disc_t^i(I) = \Omega\left(\sqrt{d \log \left(\frac{T}{d}\right)}\right)$.
\end{proof}

%----------------------------------
\subsection{Online Tusn\'ady's Problem}

Let $\boldx = (x_1, \ldots, x_T)$ be a sequence of points in $[0,1]^d$ and let $\chi \in \{\pm 1\}^T$ be a signing. For any axis-parallel box $B \subseteq [0,1]^d$ and any time $t \in [T]$, recall that  the discrepancy of axis-parallel box $B$ at time $t$ is denoted
\[ \disc_t(B,\boldx,\chi) := \Big| \chi(1) \cdot \ind_B(x_1) + \ldots +\chi(t) \cdot \ind_B(x_t) \Big| .\]
We will just write $\disc_t(B)$ when the input sequence and signing is clear from the context.

\subsubsection{Upper Bounds}

As in the interval case, it will we sufficient to work with dyadic boxes. Recall that $I_{j,k} = [k2^{-j}, (k+1)2^{-j})$ for $j \in \RZ_{\ge 0}$ and $0\le k <2^j$. To maintain the discrepancy of all intervals, it will suffice to bound the discrepancy of every dyadic box 
\[ B_{\bj,\bk} := I_{\bj(1),\bk(1)} \times \ldots \times I_{\bj(d),\bk(d)},\] 
with $\bj,\bk \in \RZ^d$ with $0\le \bj$ and $0 \le \bk < 2^{\bj}$
with each side length at least $1/T$. In particular, let $\CD = \{ B_{\bj,\bk}  \mid  0 \le \bj \le (\log T) \ind ~,~ 0 \le \bk <2^{\bj}\}$ where $\ind \in \BR^d$ is the all ones vector. Then, we prove the following lemma to bound the discrepancy of every dyadic box.

%We first give an online algorithm that maintains a small discrepancy 
%We first give an online algorithm that maintains discrepancy $O(d \log^2 T)$ with high probability for every dyadic interval along any coordinate direction $i \in [d]$, which gives us the final discrepancy bound,  
%\[ \max_{i \in [d]} \disc^i_t(I) \le O(d\log^3 T). \]

\begin{lemma} 
    \label{lemma:tusnady_main}
    Given any sequence $x_1, \ldots, x_T$ sampled independently and uniformly from $[0,1]^d$, there is an online algorithm that chooses a signing such that w.h.p. for every time $t \in [T]$,
     \[ \disc_t(B) = O\left(\log^{d+1} T\right), \text{ for all } B \in \CD.\]  
\end{lemma}

%\begin{proposition}
%    \label{prop:int_truncation}
%    Let $x_1,\ldots, x_T \sim_u [0,1]^d$ and let $\chi \in \{\pm1\}^T$ be such that $\max_{i \in [d]} \disc^i_t(I_{j,k}) \le D$ for every dyadic interval $I_{j,k}$ where $0 \le j \le \log T$ and $0 \le k < 2^{j}$ and every time $t \in [T]$. Then, with high probability over the input sequence, $\max_{i \in [d]} \disc^i_t(I) \le O((D + d)\log T)$ for every interval $I$ at every time $t$.
%\end{proposition}

Before proving Lemma~\ref{lemma:tusnady_main}, we first show why it implies Theorem~\ref{thm:tusnady}.

\begin{proof}[Proof of the upper bound in Theorem \ref{thm:tusnady}]

%To bound the discrepancy of an arbitrary interval $I$ along a coordinate axis $i \in [d]$, it suffices to get a bound on the discrepancy of every interval of the form $[\frac{a}{T}, \frac{b}{T}]$ for $a, b \in \{0,\ldots, T\}$ along that axis.
     Without loss of generality, it suffices to consider axis-parallel boxes $B = I_1 \times \cdots \times I_d$ where $I_j$'s  are half-open sub-intervals of $[0,1]$. Recall that every half-open interval $I \subseteq [0,1]$ can be decomposed as a union of at most $2\log T$ disjoint dyadic intervals in $\CD$ and two intervals $I' \subseteq I_{\log T,k}$ and $I'' \subseteq I_{\log T,k'}$ for some $0 \le k,k' < T$ (note that the length of $I'$ and $I''$ is at most $2^{-\log T} = 1/T$). 
    
    From this, it follows that for any axis-parallel box $B$, there exists a set of dyadic boxes $\CD' \subseteq \CD$ of size $|\CD'|=(2\log T)^d$ and a set $\CI$ of size $|\CI|=2d$ of disjoint intervals of length at most $1/T$, such that $B$ can be decomposed as the union of boxes in $\CD'$ and some other boxes of the form $I'_1 \times \cdots \times I'_d$, where $I'_i \in \CI$ for at least one $i \in [d]$. 
    We can therefore bound, 
    \[ \disc_t(B) \le (2\log T)^d \cdot \left(\max_{B \in \CD} \disc_t(B)\right) + N,\]
    where $N$ is the number of points $z$ in the input sequence such that $z(i) \in I$ for some $i \in [d]$ and $I \in \CI$.
    
    Applying the algorithm from Lemma \ref{lemma:tusnady_main}, the discrepancy of every dyadic box can be bounded by $O(\log^{d+1} T)$ with high probability. Also, since the length of every interval in $\CI$ is at most $1/T$, for $z$ drawn uniformly from $[0,1]^d$, we have
    \[ \BP_{z}\Big[ \exists i \in [d], \exists I \in \CI \text{ such that } z(i) \in I\Big] \le \frac{2d^2}{T}.\]
Therefore, we have that $\BE[N] \le 2d^2$ and  applying Chernoff bounds, it follows that with probability at least $1-T^{-4}$, the number $N \le 4d^2\log T$.
    
 %   the number of points such that $x_t(i) \in I_1 \cup I_2$ for any $i \in [d]$ is at most $4d\log T$. Now, using Lemma \ref{lemma:int_main} the discrepancy of every dyadic interval is bounded by $d\log^2T$ with high probability. 

Overall, with high probability for any axis-parallel box $B$, we have
    \[ \disc_t(B) ~~\le~~ (2 \log T)^d (\log^{d+1} T) + 4d^2\log T ~~=~~ O_d(\log^{2d+1} T). \qedhere \]
\end{proof}

Next, we prove the missing Lemma~\ref{lemma:tusnady_main}.
\begin{proof}[Proof of Lemma~\ref{lemma:tusnady_main}]
    We will consider the $d$-dimensional interval discrepancy problem as a vector balancing problem in $\CH_{\le \log T}^{\otimes d}$ dimensions where $\CH_{\le \log T}$ are the Haar wavelet functions with scale parameter at most $\log T$. Note that $|\CH_{\le \log T}|=T$, so the update vector in the vector balancing version will be $T^d$-dimensional. Let us abbreviate $\CH' = \CH_{\le \log T}$ and also recall that for any $\bh = (h_1,\ldots, h_d)$ in $\CH'^{\otimes d}$, we view it as a function from the cube $[0,1]^d$ to $\BR$ by defining $\bh(x) = \prod_{i=1}^d h_i(x(i))$.
  % Note that we limit the frequencies $j$ to be at most $\log T$, so the smallest possible intervals we consider are of length $1/T$.

  %  The dimensions of the update vector will correspond to $i \in [d]$ and $h \in \CH_{\le \log T}$ (note that $|\CH_{\le \log T}|=T$). 
    
    At any time when the point $x_t \in [0,1]^d$ arrives, then the $\bh := (h_1,\ldots,h_d)$ coordinate of the update vector $v_t \in [-1,1]^{\CH'^{\otimes d}}$ is given by
    \[ v_t(\bh) ~~=~~ \bh(x_t) ~~=~~ \prod_{i=1}^d h_i(x_t(i)).\] 

    We will apply the online algorithm given by Lemma \ref{lemma:main}. Let us check first that we satisfy the conditions of that lemma. First, note that the $\infnorm{v_t} \le 1$ and the vector $v_t$ has at most $(\log T + 1)^d$ non-zero coordinates, since for any fixed scale $0\le j\le \log T$ and any point $z \in [0,1]$, all but one of the values $\{h(z)\}_{h \in \CH_j}$ are zero. The last condition to check is that the coordinates of the vector $v_t$ are uncorrelated. This follows from the orthogonality of $\bh$ and $\bh'$. In particular, if $\bh \neq \bh'$, then
    \[ \BE_{v_t}[v_t(\bh)v_t(\bh')] ~~=~~ \BE_{x_t}[\bh(x_t)\bh'(x_t)] ~~=~~ 0.\]
    
    Applying the online algorithm from Lemma \ref{lemma:main}, we select signs $\chi_1,\ldots,\chi_T$ such that we get an $\ell_{\infty}$ bound on the vector $d_t = \chi_1 v_1 + \ldots \chi_t v_t$. In particular, with high probability
    \[ |d_t(\bh)| ~~=~~ \Big|\sum_{l\le t} \chi_l \bh(x_l)\Big| ~~=~~ O(\log^{d+1} T) \text{ for any } \bh \in \CH'^{\otimes d}.\]

    To finish the proof, we next bound the discrepancy of every dyadic box in terms of $\infnorm{d_t}$. For any dyadic box $B \in \CD$, since each side consists of dyadic interval $I_{j,k}$ where $j \le \log T$, Proposition \ref{prop:indhaarbox} implies that $\widehat{\ind}_{I}(\bh)=0$ for any $\bh \notin \CH'^{\otimes d}$. Therefore, we have
   \begin{align*}
       \ \disc_t(B) &= \Big|\sum_{l \le t} \chi_l \ind_{B}(x_l)\Big| ~~=~~ \Big|\sum_{l \le t} \chi_l \sum_{\bh \in \CH'^{\otimes d}} \widehat{\ind}_{B}(\bh) \bh(x_l)\Big|\\
        \ &= \Big|\sum_{\bh \in \CH'^{\otimes d}} \widehat{\ind}_{B}(\bh) \Big(\sum_{l \le t} \chi_l \bh(x_l)\Big)\Big| ~~=~~ \Big|\sum_{\bh \in \CH'^{\otimes d}} \widehat{\ind}_{B}(\bh) d_t(\bh)\Big| \\
        \                &\le  \infnorm{d_t} \cdot \Big(\sum_{\bh \in \CH'^{\otimes d}} |\widehat{\ind}_{B}(\bh)|\Big) ~~\le~~ \infnorm{d_t} ~~=~~ O(\log^{d+1} T),  
    \end{align*}
      where the second last inequality follows again from Proposition \ref{prop:indhaarbox}. 
\end{proof}

\subsubsection{Lower Bounds}

\begin{proof}[Proof of the lower bound in Theorem \ref{thm:tusnady}]

    Set $A = T^{1/d}$. We will consider the $d$-dimensional interval discrepancy problem as a vector balancing problem in $\CH_{\le \log A}^{\otimes d}$ dimensions where $\CH_{\le \log A}$ are the Haar wavelet functions with scale parameter at most $\log A$. Note that $|\CH_{\le \log A}|=A$, so the update vector in the vector balancing version will be $A^d$-dimensional. Let us abbreviate $\CH' = \CH_{\le \log A}$ and also recall that for any $\bh = (h_1,\ldots, h_d)$ in $\CH'^{\otimes d}$, we view it as a function from the cube $[0,1]^d$ to $\BR$ by defining $\bh(x) = \prod_{i=1}^d h_i(x(i))$.
  % Note that we limit the frequencies $j$ to be at most $\log T$, so the smallest possible intervals we consider are of length $1/T$.

  %  The dimensions of the update vector will correspond to $i \in [d]$ and $h \in \CH_{\le \log T}$ (note that $|\CH_{\le \log T}|=T$). 
    
    At any time when the point $x_t \in [0,1]^d$ arrives, then the $\bh := (h_1,\ldots,h_d)$ coordinate of the update vector $v_t \in [-1,1]^{\CH'^{\otimes d}}$ is given by
    \[ v_t(\bh) ~~=~~ \bh(x_t) ~~=~~ \prod_{i=1}^d h_i(x_t(i)).\] 

    We will apply Lemma \ref{lemma:main_lb}. Let us check first that we satisfy the conditions of that lemma. Similar to the proof of Lemma \ref{lemma:tusnady_main}, we note that the vector $v_t$ has exactly $(\log A + 1)^d$ non-zero coordinates that take the value $\pm1$. This implies that that Euclidean norm of any update $v_t$ is $(\log A + 1)^{d/2}$. Also from the orthogonality of $\bh$ and $\bh'$, the coordinates of the vector $v_t$ are uncorrelated --- if $\bh \neq \bh'$, then
    \[ \BE_{v_t}[v_t(\bh)v_t(\bh')] ~~=~~ \BE_{x_t}[\bh(x_t)\bh'(x_t)] ~~=~~ 0.\]

    %since for any fixed scale $0\le j\le \log A$ and any point $z \in [0,1]$, all but one of the values $\{h(z)\}_{h \in \CH_j}$ are zero. The last condition to check is that the coordinates of the vector $v_t$ are uncorrelated. 

    Applying Lemma \ref{lemma:main_lb} tells us that with probability at least $3/4$, there exists a time $t \in [T]$ and a $\bh \in \CH'$ such that $|d_t(\bh)| = \Omega \left( (\log A + 1)^{d/2}\right)$. Note that since $\bh(x) = \prod_{i=1}^d h_i(x(i))$ and $h_i$ can always be expressed as $\ind_{I_i}$ or $\ind_{I_i} - \ind_{I'_i}$ for some intervals $I_i$ and $I'_i$, it follows that there exists a set $\CB$ of at most $2^d$ axis-parallel boxes and  some $\chi \in \{\pm 1\}^{\CB}$ such that 
    \[ d_t(\bh) = \sum_{B \in \CB} \chi_B \cdot \disc_t(B).\]
    By averaging, it follows that there is an axis-parallel box $B \in \CB$ such that $\disc_t(B) \ge \dfrac{|d_t(\bh)|}{2^d}$.
    
Substituting $A=T^{1/d}$, we get that for some box $B$, 
    \[ \disc_t(B) ~~=~~ \Omega\Big(\frac{1}{2^d} \cdot \log^{d/2} A \Big) ~~=~~ \Omega_d(\log^{d/2} T). \qedhere \]
\end{proof}

\section{Applications to Online Envy Minimization}
\label{sec:envy}

In this section we use our vector balancing and two-dimensional interval discrepancy results to bound online envy. Let us first give the formal definition of envy.

Recall that there are two players and $T$ items where for item $t \in \{1, \ldots, T\}$, the valuation of the player $i\in \{1,2\}$ is $v_{it} \in [0,1]$. The \emph{cardinal envy} is the standard notion of envy studied in fair division, which is the max over every player the difference between the player's valuation for the other player's allocation and the player's valuation for their own allocation~\cite{LiptonMMS-EC04,Budish-Journal11}. Formally, if Player~$i$ is allocated set $S_i$ by an algorithm, the cardinal envy is defined as 
\[ \envy_C(\valuation_1, \valuation_2,S_1, S_2) := \max \Big\{ \sum_{t\in S_2} v_{1t} - \sum_{t\in S_1} v_{1t} ~,~ \sum_{t\in S_1} v_{2t} - \sum_{t\in S_2} v_{2t}  \Big\}.
\]

The notion of ordinal envy is defined ignoring the precise item valuations, but only with respect to the relative ordering of the items. 
Roughly, it is the worst possible cardinal envy for $[0,1]$ valuations consistent with any given relative ordering.
Thus for valuations in $[0,1]$  the ordinal envy is always at least the cardinal envy~\cite{JiangKS-arXiv19}. For $i\in \{1,2\}$, let $\pi_i$ denote the decreasing order with respect to the valuations $v_{it}$. Denote $\pi_i^t$ the first $t$ items in the order $\pi$. If Player~$i$ is allocated set $S_i$, the ordinal envy is defined as 
\[ \envy_O(\pi_1, \pi_2,S_1, S_2) := \max_{t\geq 0} \Big\{  | S_2 \cap \pi_1^{ t} | - | S_1 \cap \pi_1^{t} |  ~,~   | S_1 \cap \pi_2^{ t} | - | S_2 \cap \pi_2^{t} |   \Big\}.
\]
Jiang et al.~\cite{JiangKS-arXiv19} discuss three equivalent definitions of ordinal envy.

Next, we prove Corollary~\ref{cor:envyBounds}, which is restated below.

\envyCorol*

\begin{proof}
When the player valuations are drawn independently in $[0,1]$, the ``moreover'' part  is immediate  from the following lemma of~\cite{JiangKS-arXiv19} along with our Theorem~\ref{thm:interval} for 2-dimensional interval discrepancy.

\begin{lemma}[Lemma~26 in~\cite{JiangKS-arXiv19}] For two players with independent valuations, any upper bound for $2$-dimensional interval discrepancy problem also holds for $2$-player online ordinal envy minimization.
\end{lemma}

Next, we  bound online cardinal envy under arbitrary distributions. In the following lemma we reduce this problem to $2$-dimensional vector balancing.

\begin{lemma} For two players taking values from an arbitrary distribution $\boldp$ over $[0,1]\times [0,1]$, any upper bound for $2$-dimensional vector balancing problem also holds for $2$-player online cardinal envy minimization.
\end{lemma}
\begin{proof}
    For $i\in \{1,2\}$, let $u_{it}$ denote the valuation of Player~$i$ for $t^{th}$ item. We define the corresponding vector $v_t = (u_{1t},-u_{2t})$. If our online vector balancing algorithm assigns the next vector $v_t$ a + sign, we give the item  to Player~$2$, and otherwise we give it to Player~$1$. The crucial observation is that $d_t(1)$ and $d_t(2)$ capture precisely the cardinal envy of Players~$1$ and $2$, respectively. Thus, any bound $\| d_t\|_{\infty}$ implies a bound on the maximum cardinal envy.
\end{proof}

The last lemma when combined with Theorem~\ref{thm:signedseries} finishes the proof of  Corollary~\ref{cor:envyBounds}.
\end{proof}

\section{Open Problems and Directions}
\label{sec:open}

%\nbnote{Rewrote this section substantially. Pls check.}

 We close this paper by mentioning some 
interesting open problems  that seem to require fundamental new techniques, and new directions in online discrepancy that remain unexplored.
%\subsubsection*{Matching the offline bound}
\paragraph{Improving the dependence on $n$ for general distributions.}
Theorem \ref{thm:signedseries} gives a bound of $O(n^2\log T)$ for online vector balancing problem under inputs sampled from an arbitrary distribution. 
However, an optimal dependence of  $O(n^{1/2})$ on $n$ is achievable in the special case where the distribution has independent coordinates \cite{BansalSpencer-arXiv19}, and also in the offline setting with worst-case inputs \cite{B12}. This motivates the following question.

%The first question is if one could improve the bound here to match the offline case.
\medskip

\textbf{Question 1.} \emph{Given an arbitrary distribution $\boldp$ supported over $[-1,1]^n$, is there an online algorithm that maintains discrepancy $\sqrt{n}\cdot \polylog(T)$  on a sequence of $T$ inputs sampled i.i.d. from $\boldp$?}
\medskip

As the anti-concentration bound in Theorem \ref{lemma:anticonc} for uncorrelated variables is a $n^{1/2}$ factor worse than that for independent random variables, even getting a dependence of $ n\cdot \polylog(T)$ is an interesting first step.
%
% \noindent Note that the above bound is not obvious even if the distribution has independent coordinates, but as mentioned before, Bansal and Spencer \cite{BansalSpencer-arXiv19} give an online algorithm matching the above bound. 
%\subsubsection*{Sparse distributions}
%
\paragraph{Bounds in terms of sparsity.}
Several natural problems such as the  $d$-dimensional interval discrepancy and $d$-dimensional Tusn\'{a}dy's problem are best viewed as vector balancing problems where the input vectors are sparse.
This motivates the following online version of the Beck-Fiala problem, where the online sequence $x_1,\ldots,x_T$ is chosen independently from some distribution $\boldp$ supported over $s$-sparse $n$-dimensional vectors over $[-1,1]^n$. In the offline setting with worst-case inputs (and where we care about the discrepancy of every prefix), the methods of Banaszczyk \cite{B12} give a bound of $(s\log T)^{1/2}$.

% let $\CV$ be a universe with $T$ elements and let $\FS \subseteq 2^{\CV}$ be a family of sets where every element $v \in \CV$ is present in at most $s$ of these sets. 
% Consider the online sequence $x_1, \ldots, x_T$  where each element is chosen uniformly at random from the universe $\CV$. The goal is to bound the discrepancy of the set system $\FS$ under online signings.

%Much work has gone into studying the offline version of this problem.
%Beck \cite{Beck-Combinatorica81} proved that in the offline case, the discrepancy is $O(s)$. 
%The well-known Beck-Fiala conjecture \cite{BeckFiala-DAM81} asks whether the discrepancy bound can be improved to $O(\sqrt{s})$ which is tight. 

%The work of Banaszczyk \cite{Banaszczyk-Journal98} (see also \cite{BansalG17}) mentioned in the introduction also implies that the discrepancy is at most $O(\sqrt{s \log T})$ for a worst-case instance. Recently, there also has been a series of works \cite{HobergRothvoss-SODA19, BansalMeka-SODA19, Franks19, EzraLovett-Random19, P18} that study the stochastic version of this question in the \emph{offline} setting. The following is a natural generalization of this question in the stochastic online setting.
 
 \medskip
\textbf{Question 2.} \emph{Given an arbitrary distribution $\boldp$ supported over $s$-sparse vectors in $[-1,1]^n$, is there an online algorithm that maintains discrepancy $\poly(s,\log T, \log n)$ on a sequence of $T$ inputs sampled i.i.d. from $\boldp$?}
\medskip

 Resolving the above question would  imply polylogarithmic bounds for Tusn\'ady's problem in $d$-dimensions (similar to that in Theorem \ref{thm:tusnady}) in the much more general setting where the points $x_T$ are sampled from   an arbitrary distribution over points in $[0,1]^d$. Currently, Theorems \ref{thm:interval} and \ref{thm:tusnady} only hold when the points $x_t$ are sampled from a product distribution on $[0,1]^d$.

\paragraph{Prophet model.} The last decade has seen several online problems being studied in the prophet model where the online inputs are sampled independently from known  \emph{non-identical} distributions (see, e.g.,~\cite{Lucier-SIGecom17}). The model clearly generalizes the i.i.d. model and for point mass distributions it captures the offline problem. This model becomes useful for online problems where the adversarial arrival guarantees are weak, which raises the following question.

%Consider the following variant of online vector balancing where the input vectors are sampled independently  from \emph{non-identical} distributions: vectors $v_1,v_2,\ldots,v_T  \in [-1,1]^n$ are  drawn independently one-by-one from given distributions $\boldp_1, \boldp_2, \ldots, \boldp_T$ and presented to an online algorithm, which on arrival of $v_t \sim \boldp_t$ must irrevocably chose a sign $\chi_t \in \{\pm 1\}$, so that the $\ell_\infty$-norm of the signed sum $d_t = \chi_1 v_1 + \ldots + \chi_t v_t$  remains as small as possible. This model is inspired from the 

\medskip
\textbf{Question 3.} \emph{Given arbitrary distributions $\boldp_1, \ldots, \boldp_T$ supported over vectors in $[-1,1]^n$, is there an online algorithm that maintains discrepancy $\poly(n,\log T, \log n)$ on a sequence of $T$ inputs where vector $v_t$ is sampled independently from $\boldp_t$?}

\medskip

The techniques in Theorem~\ref{thm:signedseries} do not work since the eigenbasis may change with each arrival. It will be also interesting to study this prophet model for distributions over $s$-sparse vectors.

\paragraph{Oblivious adversary model.}
%Much work has gone into studying the offline version of this problem. Beck \cite{Beck-Combinatorica81} proved that in the offline case, the discrepancy is $O(s)$. The well-known Beck-Fiala conjecture \cite{BeckFiala-DAM81} asks whether the discrepancy bound can be improved to $O(\sqrt{s})$ which is tight. The work of Banaszczyk \cite{Banaszczyk-Journal98} (see also \cite{BansalG17}) mentioned previously also implies that the discrepancy is at most $O(\sqrt{s \log T})$ for a worst-case instance. Recently, there has been a series of works \cite{HobergRothvoss-SODA19, BansalMeka-SODA19, Franks19, EzraLovett-Random19, P18} that study the stochastic version of this question in the \emph{offline} setting.

%\subsubsection*{Oblivious Adversary Model}
A very interesting direction that is strictly harder than the above stochastic settings is to understand online vector balancing when the adversary is oblivious or non-adaptive, i.e., the adversary chooses the entire input sequence
(without any stochastic assumptions) beforehand and is \emph{not allowed} to change the inputs later based on the execution of the algorithm. 

Recall that if the adversary is fully adaptive, then one cannot hope to prove a bound better than $\Theta({T}^{1/2})$, but this might be possible for oblivious adversaries.

\medskip
\textbf{Question 4.} \emph{Is there an online algorithm that maintains discrepancy $\poly(n,\log T)$ on any sequence of $T$ vectors in $[-1,1]^n$ chosen by an oblivious adversary?}

One could also consider the same question in the Beck-Fiala setting, and ask if better bounds are possible when there is sparsity.

\medskip
\textbf{Question 5.} \emph{Is there an online algorithm that maintains discrepancy $\poly(s,\log T,\log n)$ on any sequence of $T$ vectors in $[-1,1]^n$ that are $s$-sparse and chosen by an oblivious adversary?}
\medskip

%As a first step, it would be quite interesting to prove a $\polylog(nT)$ discrepancy bound for geometric discrepancy problems like $d$-dimensional interval discrepancy or Tusn\'ady's problem in the oblivious adversary setting. 
Resolving Questions $4$ and $5$ would also have implications for both  online geometric discrepancy  and online envy minimization problems in the oblivious adversary setting.

%\nbnote{decide to ditch random order since it was mostly about $n^{1/2}$ vs $n^{1/2} \log^{1/2} n$.} 

%\nbnote{Is there something special about cardinal envy. Would this not be interesting for any problem in oblivious model? Should we mention that this is interesting even for specific problems like Tusnady's?} 

\subsection*{Acknowledgements}
We are thankful to Janardhan Kulkarni for several discussions on this project. We are also grateful to anonymous referees of STOC 2020 for their helpful comments on improving the presentation of the paper. 

\appendix

\section{Tight example for Anti-Concentration in the Original Basis for Interval Discrepancy}
\label{app:example}

%Suppose we perform a random walk at a complete binary tree where each internal node $v$ has a real number $a_v$ that satisfy $\sum_{\text{$u$ is a child of $v$}}a_u = 0$. On a random root-leaf path $P_\ell$, at node $v$ if we take a left we get $-a_v$ and if we take a right we get a $+a_v$. The goal is to find the minimum $\beta$ such that 
%\begin{align}  \label{eq:sepIneq}
%\E_{P_\ell} \Big[ \Big| \sum_{u \in P_{\ell}} a_u  \Big| \Big] ~~\geq~~ \frac{1}{\beta} \cdot \E_{P_\ell} \Big[ \sum_{u \in P_{\ell}} | a_u |  \Big] .
%\end{align}

Let us briefly recall the setting. Consider the complete binary tree of height $\log T$ where the nodes are the dyadic intervals $I_{j,k}$ for $0 \le j \le \log T$ and $0\le k < 2^j$. Our objective was to find the smallest $\beta$ such that 
\begin{align}  \label{eq:sepIneq1}
\E_x \Big[ \Big| \sum_{j,k} \sinh(\lambda d_{j,k}) \cdot \ind_{I_{j,k}}(x)  \Big| \Big] ~~\geq~~ \frac{1}{\beta} \cdot \E_x \Big[ \sum_{j,k} \cosh(\lambda d_{j,k})\cdot \ind_{I_{j,k}}(x)  \Big],
\end{align}
where $x$ is a uniform point on the unit interval $[0,1]$, the function $\ind_{I_{j,k}}$ is the indicator for the dyadic interval $I_{j,k}$, and $\lambda >0$ and $d_{j,k} \in \BR$. For simplicity, we set $\lambda = 1$ henceforth.

Observe that when a uniform random point $x$ arrives at a leaf dyadic interval $I_{\log T,k}$, then only the variables along that root-leaf path contribute to both sides. Moreover, since $x$ is uniform, the chosen leaf interval is also uniform among the leaves. Therefore, denoting by $\ell$ the random leaf and $P_{\ell}$ the corresponding root-leaf path, we want to ask for the smallest $\beta$ satisfying  
\begin{align}  \label{eq:sepIneq}
\E_{P_\ell} \Big[ \Big| \sum_{I_{j,k} \in P_{\ell}} a_{j,k}  \Big| \Big] ~~\geq~~ \frac{1}{\beta} \cdot \E_{P_\ell} \Big[ \sum_{I_{j,k} \in P_{\ell}} | a_{j,k} |  \Big],
\end{align}
where $a_{j,k} = \sinh(d_{j,k})$ for a node $I_{j,k}$ in the dyadic tree. Note that to get \eqref{eq:sepIneq} from \eqref{eq:sepIneq1}, we made the standard approximation that $\cosh(x)\approx|\sinh(x)|$ for $x \in \BR$. 

The following lemma shows that in general $\beta$ could be exponentially large in the height of the tree, so in the above case since the height is $\log T$, the value of $\beta = \Omega(\poly(T))$. We remark that for non-binary trees, this was already shown by Jiang, Kulkarni, and Singla~\cite{JiangKS-arXiv19}.

\begin{lemma} There exists $d_{j,k}$ for $0 \le j \le h$ and $0 \le k < 2^j$, such that  $\beta = \exp(\Omega(h))$ in \eqref{eq:sepIneq}.
\end{lemma}
\begin{proof}
Our construction has a fractal structure. Let $d > 0$ be a sufficiently large integer. Let $\mathcal{T}$ denote the tree structure shown in Figure~\ref{fig:fractal}(a) where the labels are the values that will be used for constructing $d_{j,k}$'s. We embed this structure in the complete binary tree of dyadic intervals and assign the $d_{j,k}$ values as follows: the root interval has value $d_{0,0}=d$ and its left children has the structure $\CT$ with the values $d_{j,k}$ as assigned by the corresponding labels in $\CT$, while the right child has value $d_{1,1}=2d/3$ and has two child subtrees with structure $\CT$ (see Figure~\ref{fig:fractal}(b)). The $d_{j,k}$ values for all the unassigned nodes (these lie in the subtree rooted at the nodes having values $d_{j,k}=-d$) are taken to be zero.

%where the values corresponding t the following values for $d_{j,k}$'s  
%\begin{itemize}
%\item $d_{0,0} = d/3 $.
%\item $d_{1,0} = -d/3$ and $d_{1,1} = +2d/3$.
%\item $d_{2,0}  = -d$ and $d_{2,1} = +2d/3$ and $d_{2,2} = \mathcal{T}$ and $d_{2,3} = \mathcal{T}$.
%\item $d_{j,k}$=0 for any node in the subtree rooted at $I_{3,0}$ and $I_{3,1}$, while $d_{3,2} = \mathcal{T}$ and $d_{3,3} = \mathcal{T}$.
%\end{itemize}

\begin{figure}
   \centering
   \subfloat[Fractal structure $\CT$]{{\includegraphics[width=0.4\textwidth]{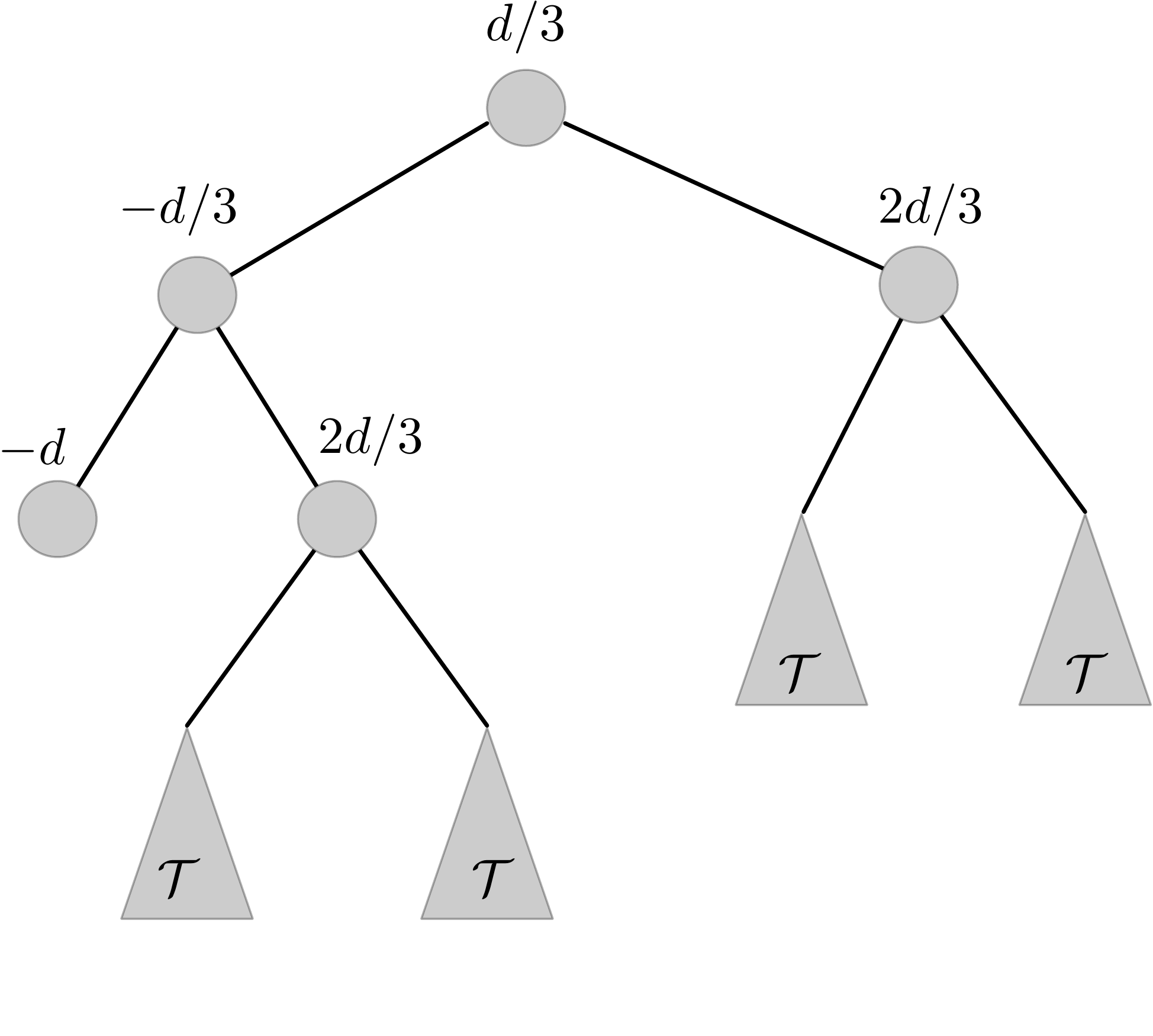}}}%
    \qquad\qquad
    \subfloat[Embedding of $\CT$ in the dyadic tree]{{\includegraphics[width=0.36\textwidth]{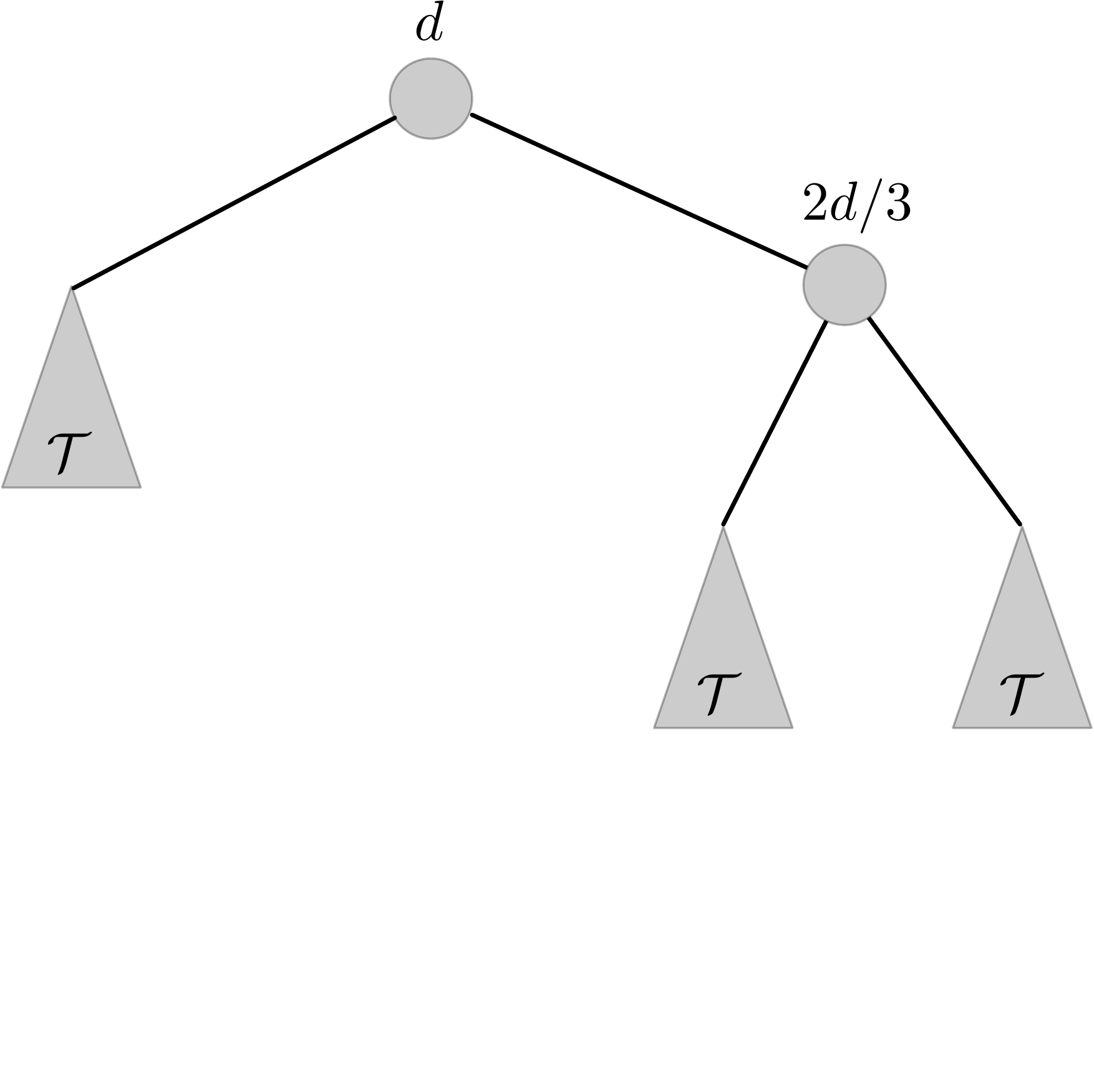}}}%
    \caption{Construction of $d_{j,k}$'s satisfying \eqref{eq:sepIneq}}%
    \label{fig:fractal}%
\end{figure}
   
%To prove the lemma, consider a tree with the root $r$ having $a_r = \sinh(d)$.  It has one child with $\mathcal{T}$ and another child with $\sinh(2d/3)$. The second child has two subtrees with structure $\mathcal{T}$. 

Note that $\mathcal{T}$ has the property that with probability $1/4$ it ends in a node $I_{j,k}$ with $a_{j,k}=\sinh(-d)$, and otherwise it enters another $\mathcal{T}$ (unless we already reached a leaf).

The proof now follows because if we take a random root-leaf path in our dyadic tree, with probability $1 - \exp(-\Omega(h))$ it will end in a leaf with $\sinh(-d)$, which will cancel with $\sinh(d)$ at the root. Since every other entry on a root leaf path  has magnitude at most $\sinh(2d/3)$,
the left hand side in \eqref{eq:sepIneq} will be 
\[ \E_{P_\ell} \Big[ \Big| \sum_{I_{j,k} \in P_{\ell}} a_{j,k}  \Big| \Big] ~~\leq~~  \Big( 1 - \exp(- \Omega(h)) \Big) \cdot h \cdot |\sinh(2d/3)| + \exp(-\Omega(h)) \cdot h \cdot |\sinh(d)| ~~\leq ~~ \frac{|\sinh(d)|}{\exp(\Omega(h))}, 
\]
while the right hand side is 
\[ \E_{P_\ell} \Big[ \sum_{I_{j,k} \in P_{\ell}} | a_{j,k} |  \Big] ~~\ge~~ |\sinh(d)|.\]

Therefore, $\beta = \exp(\Omega(h))$ in \eqref{eq:sepIneq}.
\end{proof}

\section{Burkholder-Davis-Gundy Inequality}
\label{sec:bdg}

Let $Z_0, Z_1, \ldots, Z_t$ be a discrete martingale (with respect to $W_1, \ldots, W_t$) and let $\Delta Z_s = Z_s - Z_{s-1}$ denote the differences for all $s \in [t]$. Note that $Z_s = \Delta Z_1 + \Delta Z_2 + \ldots + \Delta Z_s$. Define $Z_t^* = \max_{0 \le s \le t} |Z_s|$ to be the maximum value of the martingale process till time $t$. Then, the well-known  Burkholder-Davis-Gundy inequality says the following.

\begin{theorem}[\cite{BDG72}]
\label{thm:bdg}
Let $1\le p < \infty$. Then, there exist positive constants $c_p$ and $C_p$ such that 
\[ c_p\cdot\BE\left[ \left(\sum_{s=1}^t |\Delta Z_s|^2\right)^{p/2} \right] ~~\le~~ \BE[(Z_t^*)^p] ~~\le~~ C_p\cdot \BE\left[ \left(\sum_{s=1}^t |\Delta Z_s|^2\right)^{p/2} \right]. \] 
\end{theorem}

Note that the inequality holds in much more general settings, but the above setting is sufficient for the purposes of this paper.

Furthermore, for $p=1$, which is the case we need for the purposes of this paper, one can relate expected magnitude of $Z_t^*$ and $Z_t$ by the following inequality.

\begin{lemma} 
\label{lemma:lstar}
$\BE[Z_t^*] \le (t+1) \cdot \BE[|Z_t|]$ .
\end{lemma}
\begin{proof}
First note that $f(Z_0), \ldots, f(Z_t)$ is a sub-martingale with respect to $W_1, \ldots, W_t$ for any convex function $f$. Choosing $f(z)=|z|$, it follows that the absolute value of the above martingale is a sub-martingale. Applying Doob's optional stopping theorem to this sub-martingale, one gets that $\BE[|Z_t|] \ge \BE[|Z_0|]$. Since, we could have started this sequence anywhere, it also follows for any $s < t$ that $\BE[|Z_t|]\ge \BE[|Z_s|]$. 

Since $Z_t^* = \max_{s\le t} |Z_s| \le \sum_{s=0}^t |Z_s|$, using linearity of expectation, we get that
\[ \BE[Z_t^*] ~~\le~~ \sum_{s=0}^t \BE[|Z_s|] ~~\le~~ (t+1)\BE[|Z_t|]. \qedhere\]
\end{proof}

%\clearpage

{%\small
\bibliographystyle{alpha}
\bibliography{fullbib}
}

\end{document}